\documentclass[11pt]{article}

\usepackage{lmodern}
\usepackage[T1]{fontenc}

\usepackage{hyperref}
\hypersetup{
    colorlinks=true,
    linkcolor=DarkRed,
    filecolor=magenta,      
    urlcolor=cyan,
    citecolor=red,
    hidelinks
}

\usepackage{fullpage}

\usepackage{amsmath,amsthm,amsbsy,amssymb,mathtools}
\usepackage{dsfont}

\usepackage[nopatch=footnote]{microtype}
\usepackage{url}
\usepackage{booktabs}
\usepackage{multirow}
\usepackage{paralist}
\usepackage{xspace}
\usepackage{tabularx}
\usepackage{verbatim}
\usepackage{graphicx}
\usepackage{float}
\usepackage{setspace}
\usepackage{snapshot}
\usepackage{multicol}

\usepackage{xcolor}
\usepackage{color}

\usepackage{subcaption}
\usepackage{tikz}
\usetikzlibrary{
    graphs,
    graphs.standard,
    shapes,
    arrows.meta,
    positioning
}

\usepackage[ruled, lined, linesnumbered, commentsnumbered, noend]{algorithm2e}
\usepackage[noend]{algorithmic}

\usepackage{thmtools}
\usepackage{thm-restate}
\usepackage{cleveref}

\newtheorem{theorem}{Theorem}[section]

\newtheorem{lemma}[theorem]{Lemma}
\newtheorem{proposition}[theorem]{Proposition}

\newtheorem{claim}[theorem]{Claim}
\newtheorem{observation}[theorem]{Observation}

\theoremstyle{definition}
\newtheorem{definition}[theorem]{Definition}
\newtheorem{condition}[theorem]{Condition}

\newtheorem{remark}[theorem]{Remark}

\usepackage{todonotes}

\newcommand{\poly}{\textrm{poly}}

\newcommand{\eps}{\ensuremath{\varepsilon}}

\newcommand{\dout}{d^{\textrm{out}}}
\newcommand{\Dout}{\Delta^{\textrm{out}}}
\newcommand{\sgn}{\text{sgn}}
\newcommand{\E}{\mathbb{E}}
\newcommand{\Var}{\mathbf{Var}}
\providecommand{\abs}[1]{\ensuremath{\left\lvert#1\right\rvert}}
\providecommand{\norm}[1]{\ensuremath{\lVert#1\rVert}}

\title{Sublinear-Time Algorithms for Diagonally Dominant Systems and Applications to the Friedkin-Johnsen Model
}
\newcommand*\samethanks[1][\value{footnote}]{\footnotemark[#1]}
\author{Weiming Feng \thanks{School of Computing and Data Science, The University of Hong Kong, China. Supported by the ECS grant 27202725 from Hong Kong RGC. Email: \tt{wfeng@hku.hk}.}\and Zelin Li \thanks{School of Computer Science and Technology, University of Science and Technology of China, China. Supported in part by NSFC grant 62272431 and Innovation Program for Quantum Science and Technology (Grant No. 2021ZD0302901). Emails: \tt{meguru@mail.ustc.edu.cn, ppeng@ustc.edu.cn}. }\and Pan Peng \samethanks}
\date{}

\begin{document}
\maketitle
\begin{abstract}
We study sublinear-time algorithms for solving linear systems $Sz = b$, where $S$ is a diagonally dominant matrix, i.e., $|S_{ii}| \geq \delta + \sum_{j \ne i} |S_{ij}|$ for all $i \in [n]$, for some $\delta \geq 0$. 
We present randomized algorithms that, for any $u \in [n]$, return an estimate $z_u$ of $z^*_u$ with additive error $\varepsilon$ or $\varepsilon \norm{z^*}_\infty$, where $z^*$ is some solution to $Sz^* = b$, and the algorithm only needs to read a small portion of the input $S$ and $b$. For example, when the additive error is $\varepsilon$ and assuming $\delta>0$, we give an algorithm that runs in time $O\left( \frac{\|b\|_\infty^2 S_{\max}}{\delta^3 \varepsilon^2} \log \frac{\| b \|_\infty}{\delta \varepsilon} \right)$, where $S_{\max} = \max_{i \in [n]} |S_{ii}|$. We also prove a matching lower bound, showing that the linear dependence on $S_{\max}$ is optimal. Unlike previous sublinear-time algorithms, which apply only to symmetric diagonally dominant matrices with non-negative diagonal entries, our algorithm works for general strictly diagonally dominant matrices ($\delta > 0$) and a broader class of non-strictly diagonally dominant matrices $(\delta = 0)$. Our approach is based on analyzing a simple probabilistic recurrence satisfied by the solution. As an application, we obtain an improved sublinear-time algorithm for opinion estimation in the Friedkin--Johnsen model.
\end{abstract}

\section{Introduction}
Solving linear systems is a fundamental computational problem with numerous applications. Given a matrix $ A \in \mathbb{R}^{n \times n} $ and a vector $ b \in \mathbb{R}^n $, the goal is to find a solution $ z \in \mathbb{R}^n $ satisfying $ Az = b $. Numerous polynomial-time algorithms have been developed for this task. For example, researchers have established that, under exact arithmetic (RealRAM model), linear systems can be solved in $O(n^{\omega})$ time, where $ \omega < 2.371339 $ is the matrix multiplication exponent \cite{alman2025more}, through a reduction to fast matrix multiplication. When $ A $ is sparse, the conjugate gradient method applied to the normal equations $ A^\top A z = A^\top b $ solves the system in $ O(\mathrm{nnz}(A) \cdot n) $ time (see, for example, \cite{spielman2010algorithms}). Special classes of linear systems admit even faster algorithms. For instance, \emph{symmetric diagonally dominant (SDD)} matrices---defined by the condition $ a_{ii} \geq \sum_{j \neq i} |a_{ij}| $ for all $ i $---allow nearly linear-time solvers. This class includes the Laplacian matrix $ L_G $ of an undirected graph $ G $. Beginning with the seminal work of Spielman and Teng \cite{spielman2004nearly}, a series of results have demonstrated that systems like $ L_G z = b $ can be solved in nearly linear time (e.g., \cite{koutis2014approaching, cohen2014solving}). The current state-of-the-art runs in $ \tilde{O}(\mathrm{nnz}(A) \log (1/\varepsilon)) $ time, where $ \varepsilon $ is an accuracy parameter and $ \tilde{O}(\cdot) $ suppresses polylogarithmic factors. Extensions to directed graphs have also yielded nearly linear-time solvers for certain Laplacian systems (e.g., \cite{cohen2018solving}). 
%
%
%
%
Solving linear systems where the coefficient matrix is diagonally dominant (DD) also arises naturally in numerous practical settings, such as PageRank or Personalized PageRank (PPR), computing the effective resistance, and machine learning on graphs including semi-supervised learning and label propagation~\cite{10.1007/978-3-031-31975-4_19,DBLP:conf/nips/HoltzCWCM24}.

Motivated by advances in fast linear-system solvers under alternative computational models and the challenges posed by modern large-scale datasets, Andoni, Krauthgamer, and Pogrow~\cite{andoni2019solving} pioneered \emph{sublinear-time algorithms} for solving diagonally dominant linear systems. These algorithms inspect only a sparse subset of the input (the matrix $ A $ and vector $ b $) while preserving rigorous guarantees on output quality.

The work in~\cite{andoni2019solving} addresses settings where only one or a few coordinates of the solution $ z^* \in \mathbb{R}^n $ are sought. Focusing on SDD matrices $S$ with non-negative diagonal entries, their algorithm assumes query access to both $S$ and $b$, which corresponds to vertex queries and neighbor queries in the underlying graph. For any index $ u \in [n] $, it computes an estimate $ \hat{z}_u $ approximating the $ u $-th coordinate of a  solution $ z^* $ to the system $ Sz = b $. The runtime scales polynomially in $ \kappa $, an upper bound on the condition number of $ S $, and becomes sublinear in the input size when $ \kappa $ is small (see more discussions on the runtime below). This framework is applicable to a wide range of applications. Specifically, as noted in \cite{andoni2019solving}, the problem of estimating effective resistances between pairs of vertices in graphs can be reduced to estimating a few coordinates of the solution to a Laplacian linear system. The corresponding algorithm was subsequently improved by Peng et al. \cite{peng2021local}. Additionally, Neumann et al. \cite{neumann2024sublinear} observed that such sublinear-time algorithms can be applied to determine the opinions of specific nodes in social networks within the Friedkin-Johnsen (FJ) model (see further discussions below). Furthermore, the development of sublinear-time linear system solvers has been influenced by recent advancements in quantum algorithms for linear systems \cite{harrow2009quantum,ambainis2012variable,childs2017quantum,dervovic2018quantum} (see \Cref{sec:relatedwork}).

While powerful, the algorithm in \cite{andoni2019solving} is limited to SDD matrices with non-negative diagonal entries, which is a special class of positive semi-definite (PSD) matrices. Their approach involves treating the matrix $ S $ as the Laplacian matrix of an undirected graph $ G $ and performing short-length random walks on the corresponding graph. The estimates are then defined based on the outcomes of these random walks. This method inherently works only for $ S $ being symmetric and having non-negative diagonal entries.  

In this paper, we present new sublinear-time algorithms for solving linear systems as long as the matrix $ S $ is diagonally dominant, i.e., $ |S_{ii}| \geq \sum_{j \neq i} |S_{ij}| $ for all $ i $. \emph{Note that $ S $ is not necessarily symmetric or constrained to have non-negative diagonal entries. In fact, such $ S $ is not necessarily PSD.} 
We also provide a lower bound to demonstrate that our algorithm is nearly optimal. Finally, we apply our algorithm to obtain improved sublinear algorithms for opinion estimation in the Friedkin-Johnsen  model. 

\subsection{Main results}\label{sec:results}
\paragraph{The problem and computational model} 
We study the problem of estimating individual entries of the solution $z^*$ to the linear system $Sz = b$
in sublinear time, where $S \in \mathbb{R}^{n \times n}$ is a $\delta$-diagonally dominant ($\delta$-DD) matrix and $b \in \mathbb{R}^{n}$. The definition of $\delta$-DD matrix is given below.
\begin{definition}[$\delta$-Diagonally Dominant ($\delta$-DD) Matrix]
Let $\delta \geq 0$ be a non-negative parameter. A matrix $S \in \mathbb{R}^{n \times n}$ is called \emph{$\delta$-diagonally dominant} (or \emph{$\delta$-DD}) if
\begin{align*}
    \forall i \in [n], \quad |S_{ii}| \geq \delta + \sum_{j: j\neq i}
    |S_{ij}|.
\end{align*}
If $\delta > 0$, we say that $S$ is \emph{strictly DD}. If $\delta = 0$, we say $S$ is \emph{non-strictly DD}.
\end{definition}


Note that if $S$ is strictly DD, then it is non-singular, so the solution $z^*$ is unique.


Now we specify the query model: the algorithm accesses the information of $S$ and $b$ through queries to an oracle.
To define the oracle, we first model $S \in \mathbb{R}^{n \times n}$ as a directed graph $ G = G(S) = (V, E) $, where $V = [n]$.
For any pair $(u,v)$ with $u \neq v$ and $S_{uv} \neq 0$, there exists a  directed edge $ (u, v) \in E $ with weight $ S_{uv} $. For each vertex $ u \in [n] $, define the out-degree $ \Dout_u $ as the number of out-edges from vertex $ u $, and define the out-weighted degree as $\dout_u = \sum_{v \neq u} |S_{uv}|$. 
Note that $ \Dout_u = \dout_u $ when $ S_{uv} \in \{-1, 0, 1\} $. For $b \in \mathbb{R}^n$, the value $ b_u $ is stored at vertex $ u $.

\begin{definition}[$(S,b)$-oracle]\label{def:oracle}
Let $ S \in \mathbb{R}^{n \times n} $, $ b \in \mathbb{R}^n $, and let $ G = ([n], E) $ be the directed graph defined from $ S $. The $(S,b)$-oracle supports the following queries, each with unit query cost: 
\begin{enumerate}
    \item \textbf{Vertex query:} Given $ u \in [n] $, return $ \Dout_u $, $ \dout_u $, $ S_{uu} $, and $ b_u $. 
    \item \textbf{Neighbor query:} Given $ u \in [n] $ and $ i \in [\Dout_u] $, return the $ i $-th out-neighbor $ v $ of $ u $ (under a fixed but arbitrary ordering).
    \item \textbf{Random walk query:} Given $ u \in [n] $, return a random out-neighbor $ v \neq u $ with probability proportional to $ |S_{uv}| $, along with the edge weight $ S_{uv} $. If $u$ has no out-neighbors, then return a special symbol $\perp$. 
\end{enumerate}
\end{definition}

The vertex query and neighbor query can be viewed as a natural extension of local queries on unweighted undirected graphs to the setting of weighted directed graphs \cite{Goldreich2010}. The random walk query has been considered in more recent works \cite{conf/colt/JinKMSS24,neumann2024sublinear}.

We consider the problem of solving the linear system $Sz = b$ in the query model. Given the access to the $(S,b)$-oracle, an input vertex $u \in [n]$, an error bound $\varepsilon > 0$, and possible other additional information such as $\delta$ and $\Vert b \Vert_\infty$, the algorithm is required to compute a value $z_u$ in sublinear time such that $|z_u-z^*_u|<\varepsilon$, where $z^*$ is the unique solution satisfying $Sz^*=b$.

\subsubsection{Algorithmic results}

We first present our result on strictly diagonally dominant matrices ($\delta$-DD for $\delta > 0$).  We define $S_{\max}=\max_{i\in[n]}|S_{ii}|$.

\begin{theorem} \label{thm:algo}
Let $\delta > 0$, $S \in \mathbb{R}^{n \times n}$ be a $\delta$-diagonally dominant matrix  and  $b \in \mathbb{R}^n$.
There exists an algorithm with access to the $(S,b)$-oracle such that given the value of $\delta$, the value of $\Vert b \Vert_{\infty}$, an error bound $\varepsilon > 0$ and a vertex $u \in [n]$, the algorithm returns a random value $z_u$ with query complexity $O(\frac{\lVert b\rVert_{\infty}^2S_{\max}}{\delta^3\varepsilon^2}\log \frac{\lVert b\rVert_{\infty}}{\delta\varepsilon})$ such that with probability at least $2/3$, $z_u$ satisfies $|z_u-z^*_u|<\varepsilon$, where $z^*$ is the unique solution satisfying $Sz^*=b$.
\end{theorem}



\begin{remark}[Median trick]\label{remark:median}
    The success probability $2/3$ is an arbitrary constant greater than $1/2$. To achieve $1 - \delta$ success probability for an arbitrary $\delta > 0$, one can independently run the algorithm with  $2/3$ success probability for $O(\log \frac{1}{\delta})$ times and then take the median of all the outputs.
\end{remark}

\Cref{thm:algo}  considers the additive approximation error. 
As a corollary, we can also obtain a sublinear-time algorithm that achieves relative error with respect to $\Vert z^*\Vert_\infty$.

\begin{restatable}{corollary}{1.1}\label{cor:rel}
Let $\delta > 0$. There exists an algorithm with access to $(S,b)$-oracle for any $\delta$-diagonally dominant matrix $S$ and let $b \in \mathbb{R}^n$ be a vector such that for given the value of $\delta$, the value of $\Vert b \Vert_{\infty}$, the value of $S_{\max}$, an error bound $\varepsilon > 0$, and a vertex $u \in [n]$,
the algorithm returns a random value $z_u$ with query complexity $O\left( \frac{S_{\max}^3}{\delta^3\varepsilon^2}\log \frac{S_{\max}}{\delta\varepsilon}\right)$ such that with probability at least $2/3$, $z_u$ satisfies $|z_u-z^*_u|<\varepsilon||z^*||_{\infty}$, where $z^*$ is the solution satisfying $Sz^*=b$. 
\end{restatable}





We next consider the case that $S$ is non-strictly diagonally dominant matrices, which means for all $i \in [n]$, $|S_{ii}| \geq \sum_{j \in [n]:j \neq i}|S_{ij}|$. 
We consider matrices $S$ satisfying the following condition.
\begin{condition}\label{cond:S} $S$ is non-strictly DD and $(S,b)$ satisfies one of the following two conditions
\begin{enumerate}
    \vspace{-0.5em}
    \item $S$ is non-singular;
    \vspace{-0.5em}
    \item $S$ is symmetric, and all the diagonal entries are non-zero and have the same sign. 
\end{enumerate}
\end{condition}
\vspace{-0.5em}
Define $\infty$-norm condition number of $S$ as $\kappa_{\infty}(S)=\lVert S\rVert_{\infty}\lVert S^{-1}\rVert_{\infty}$, where $S^{-1}$ is the pseudo-inverse of $S$ if $S$ is singular\footnote{See definitions of pseudo-inverse and matrix norm $\Vert \cdot \Vert_\infty$ in \Cref{sec:preliminaries}}. 
Let $z^*$ be a solution to $Sz = b$. 
Under the first condition of in \Cref{cond:S}, $z^*$ is unique. Under the second condition, $S$ is symmetric and we define $z^* = S^{-1}b$, where $S^{-1}$ is the pseudo-inverse of the symmetric matrix $S$.

\begin{restatable}{theorem}{cond}\label{thm:cond}
There exists an algorithm with access to $(S,b)$-oracle for any $S$ satisfying \Cref{cond:S} and any $b$ in the range of $S$ such that given the values of $\kappa_\infty(S)$, $\Vert b \Vert_{\infty}$, $S_{\max}$, and a vertex $u \in [n]$,
the algorithm returns a random value $z_u$ with query complexity $O\left( \frac{\kappa_\infty(S)^3}{\varepsilon^5}\log \frac{\kappa_\infty(S)}{\varepsilon}\right)$ such that with probability at least $2/3$, $z_u$ satisfies $|z_u-z^*_u|<\varepsilon\norm{z^*}_{\infty}$, where $z^* = S^{-1}b$. 
\end{restatable}

\paragraph{Comparison to \cite{andoni2019solving}.}
We compare our results with those in \cite{andoni2019solving}, who study a similar problem under different assumptions. Let $D = \mathrm{diag}(S_{11}, \dots, S_{nn})$ and define $\tilde{S} = D^{-1/2} S D^{-1/2}$. Their algorithm outputs $z_u$ such that $|z_u - z^*_u| < \varepsilon \|z^*\|_{\infty}$ in time $\tilde{O}(\kappa^3 / \varepsilon^2)$, where $\kappa$ is the $2$-norm condition number of $\tilde{S}$ and $\tilde{O}(\cdot)$ hides polylogarithmic factors. 
In contrast, our \Cref{thm:algo} achieves absolute additive error $\varepsilon$, while \Cref{cor:rel} and \Cref{thm:cond} guarantee additive error $\varepsilon \|z^*\|_{\infty}$, which may yield different accuracies depending on $\|z^*\|_\infty$. Furthermore, for networks where degree differences are not significant, such as Small-World and Erdős–Rényi (ER) networks, in the Friedkin-Johnsen model, our algorithms achieve better query complexity. Additionally, their algorithm applies only to symmetric DD matrices with strictly positive diagonal entries, whereas our method applies to general DD matrices without such restrictions. The runtime of their algorithm depends on the condition number $\kappa$, while our runtime is characterized by different parameters. In certain cases, such as when $S = \delta I+L$ for $L$ being the Laplacian of an undirected regular graph, both algorithms achieve similar complexity up to polylogarithmic factors since $\kappa = \Theta(S_{\max}/\delta)$ and our result in \Cref{cor:rel} matches theirs. However, when $S$ is the Laplacian of an undirected regular graph (i.e., $\delta = 0$), their algorithm can be faster because $\kappa_2(\tilde{S}) \leq \kappa_\infty(S)$.

\subsubsection{Lower bound} Note that when $\norm{b}_\infty\leq 1$, then the query complexity of the algorithm from \Cref{thm:algo} is $O\left( \frac{S_{\max}}{\delta^3 \varepsilon^2} \log \frac{1}{\delta \varepsilon} \right)$. 
The following lower bound shows our algorithm has an optimal dependence on $S_{\max}$.
Consider the case that $S$ is non-singular.
We say a sublinear 
time algorithm $\mathcal{A}$ with access to $(S,b)$-oracle solves the $Sz = b$ linear system with additive error $\varepsilon$ if for any $u \in [n]$, it returns a random value $z_u$ satisfying $|z_u-z_u^*|\leq \varepsilon$ with probability at least $2/3$, where $z^*$ is the unique solution satisfying $Sz^*=b$.
Let $\Delta_{\max}^S$ denote the maximum out-degree of graph $G$ constructed from $S$, i.e., $\Delta_{\max}^S = \max_{u \in [n]}|\{ v\in [n]\mid v \neq u \land S_{uv} \neq 0\}|$.

\begin{restatable}{theorem}{thm:lb}\label{thm:lb}
There exists a small constant $\varepsilon_0 > 0$ for which the following hardness result holds:
For any sublinear-time algorithm $\mathcal{A}$ that, given a $1$-diagonally dominant matrix $S \in \mathbb{R}^{n \times n}$ satisfying $S_{\max} = \Theta(\sqrt{n})$ and $\Delta^S_{\max} = O(1)$, and any Boolean vector $b \in \{0,1\}^n$, solves the linear system $Sz = b$ with additive error $\varepsilon_0$, the algorithm $\mathcal{A}$ must make $\Omega(S_{\max})$ queries to the $(S,b)$-oracle.
\end{restatable}

\begin{remark}\label{remark:lb}
In the proof of \Cref{thm:lb}, we construct the hard instance matrix $S \in \mathbb{R}^{n \times n}$ as $S = I+L$, where $I$ is the identity matrix, $L$ is the Laplacian matrix of an undirected weighted regular graph $G_L$ with constant degree, and each edge in $G_L$ has weight $\Theta(\sqrt{n})$. 
\end{remark}

Taking $\delta = 1$, $\varepsilon = \varepsilon_0 = \Theta(1)$, and $\Vert b \Vert_\infty = 1$, the algorithm in \Cref{thm:algo} achieves query complexity $O(S_{\max})$, which matches the lower bound in \Cref{thm:lb}. For the hard instances constructed in \Cref{thm:lb}, we have $S_{\max} = \Theta(d^S_{\max})$, where $d^S_{\max} = \max_{u}\sum_{v \neq u} |S_{uv}|$ is the maximum weighted out-degree of the graph $G$ constructed from $S$. This lower bound shows that even when the maximum out-degree $\Delta^S_{\max}$ is constant, the running time must be at least linear in the maximum weighted out-degree $d^S_{\max}$. When expressed in terms of $n$, \Cref{thm:lb} establishes an $\Omega(\sqrt{n})$ lower bound on the query complexity.

\paragraph{Comparison to \cite{andoni2019solving}.}
Andoni, Krauthgamer, and Pogrow~\cite{andoni2019solving} also studied lower bounds for sublinear-time algorithms. They proved that for sparse matrices $S$ (where $\Delta^S_{\max} = O(1)$) with condition number $\kappa_2(S) \leq 3$, any algorithm that solves the linear system with additive error $\frac{1}{5}\Vert z^* \Vert_\infty$ must probe at least $n^{\Omega(1/(\Delta^S_{\max})^2)}$ entries of $b$. Their lower bound holds in a stronger computational model where the entire matrix $S$ is known to the algorithm in advance. In contrast, our work establishes an $\Omega(\sqrt{n})$ lower bound on the total number of queries required in our query model to achieve $O(1)$ additive error, where the algorithm has no prior knowledge of $S$ or $b$. 
While they~\cite{andoni2019solving} study the lower bound when $S$ is a Laplacian matrix $L$, we focus on the case $S = I+L$, which is particularly relevant for the Friedkin-Johnsen model discussed below.


\subsubsection{Applications to the Friedkin-Johnsen opinion dynamics} The
  Friedkin--Johnsen (FJ) model is a popular model for studying opinion formation processes in social networks ~\cite{friedkin1990social}. 
Neumann et al.\ applied the solver from~\cite{andoni2019solving} to the opinion estimation problem in the Friedkin–Johnsen model \cite{neumann2024sublinear}. Let $G_{\text{FJ}} = (V, E, w)$ be an undirected, weighted graph. Each node $u$ has an \emph{innate opinion} $b_u \in [0,1]$ and an \emph{expressed opinion} $z_u \in [0,1]$. While the innate opinions are fixed, the expressed opinions evolve over time $t$ according to the update rule:
    \begin{align}
        \label{eq:update-rule}
        z_u^{(t+1)} = \frac{b_u + \sum_{(u,v) \in E} w_{uv} z_v^{(t)}}{1 + \sum_{(u,v) \in E} w_{uv}},
    \end{align}
    i.e., each node updates its expressed opinion as a weighted average of its own innate opinion and the expressed opinions of its neighbors. It is known that as $t \to \infty$, the expressed opinions converge to the equilibrium $z^* = (I + L)^{-1} b$, where $I$ is the identity matrix and $L = D - A$ is the graph Laplacian of $G_{\text{FJ}}$, with $D$ as the weighted degree matrix and $A$ as the weighted adjacency matrix.

    It was observed by Neumann, Dong and Peng in~\cite{neumann2024sublinear} that $z_u^*$ can be approximated using the sublinear-time solver from~\cite{andoni2019solving}, yielding an estimate $\tilde{z}_u^*$ such that
    \[
        |\tilde{z}_u^* - z_u^*| \leq \varepsilon
    \]
    with probability $1 - \frac{1}{s}$, where $s = O(\kappa \log(\varepsilon^{-1} n \kappa \max_u w_u))$ and $\kappa$ is an upper bound on the 2-norm condition number of $(I+D)^{-\frac12}(I + L)(I+D)^{-\frac12}$. The algorithm runs in time $O(\varepsilon^{-2} s^3 \log s) = \tilde{O}(\frac{\kappa^3}{\varepsilon^2})$.
Furthermore, \cite{neumann2024sublinear} gave a \emph{deterministic} sublinear-time algorithm for estimating $z_u^*$ based on a sublinear-time algorithm for computing personalized PageRank~\cite{andersen2006local}. 
This algorithm works if $G_{\text{FJ}}$ is a regular graph and the query complexity is  $(W/\varepsilon)^{O(W\log(1/\varepsilon))}$, where $W$ is the weighted degree of the graph $G_{\text{FJ}}$.


We apply our new sublinear-time algorithm in \Cref{thm:algo} to the Friedkin-Johnsen model. We set $S = I+L$ and we define $W$ is the maximum \emph{weighted} degree of the graph $G_{\text{FJ}}$ for Friedkin-Johnsen model such that $W = \max_{u \in V} \sum_{v: \{u,v\} \in E} w_{uv}.$ 
We apply our new sublinear-time algorithm in \Cref{thm:algo} to the Friedkin-Johnsen model. We set $S = L + I$ and we define $W$ as the maximum \emph{weighted} degree of the graph $G_{\text{FJ}}$ for Friedkin-Johnsen model such that $W = \max_{u \in V} \sum_{v: \{u,v\} \in E} w_{uv}.$ 
Then, the parameter $S_{\max} = W + 1$.
Note that for Friedkin-Johnsen model, $S = I+L$ is $1$-diagonally dominant and $\Vert b \Vert_\infty \leq 1$.
We can obtain the following fast sublinear-time algorithm. 
\begin{theorem}\label{thm:fj}
For the Friedkin-Johnsen model, there exists an algorithm with access to the $(I+L,b)$-oracle such that given an error bound $\varepsilon > 0$ and a vertex $u \in [n]$, the algorithm returns a random value $z_u$ with query complexity $O(\frac{W+1}{\varepsilon^2}\log \frac{1}{\varepsilon})$ such that with probability at least $2/3$, $z_u$ satisfies $|z_u-z^*_u|<\varepsilon$, where $z^*$ is the unique solution satisfying $(I+L)z^*=b$.
\end{theorem}
Again, we can boost the success probability from $2/3$ to $1-\delta$ via \Cref{remark:median}.

\paragraph{Worst-case optimal query complexity.}
The query complexity in \Cref{thm:fj} is worst-case \emph{optimal} for the Friedkin-Johnsen model. By \Cref{remark:lb}, the hard instances $(S,b)$ in \Cref{thm:lb} satisfy $S = I + L$ and $b \in \{0,1\}^n$, where $L$ is the Laplacian matrix of an undirected weighted regular graph $G_L$. We can take $G_{\text{FJ}} = G_L$, then the hardness instance becomes a Friedkin-Johnsen model on graph $G_{\text{FJ}}$. The lower bound in \Cref{thm:lb} shows that to achieve constant error $\varepsilon = O(1)$, the query complexity is $\Omega(S_{\max}) = \Omega(W+1)$. For constant error, our algorithm in \Cref{thm:fj} achieves the optimal $O(W+1)$ query complexity.

\paragraph{Comparison to \cite{neumann2024sublinear}.}
We compare our results with those in \cite{neumann2024sublinear} both theoretically and experimentally. Consider the case where $G_{\text{FJ}}$ is a regular graph, where every edge has the same weight and the maximum weighted degree is $W$. In this case, the 2-norm condition number $\kappa$ of $(I+D)^{-\frac12}(I + L)(I+D)^{-\frac12}$ is $\Theta(W+1)$. The algorithm in \cite{neumann2024sublinear} achieves the query complexity $\tilde{O}(\frac{(W+1)^3}{\varepsilon^2})$. The deterministic algorithm in \cite{neumann2024sublinear} achieves the query complexity $(W/\varepsilon)^{O(W\log(1/\varepsilon))}$. Our algorithm achieves a much better query complexity $\tilde{O}(\frac{W+1}{\varepsilon^2}\log \frac{1}{\varepsilon})$.

Furthermore, we conduct experiments on the Friedkin--Johnsen model and compare the performance of our algorithm with that of \cite{neumann2024sublinear} on the same dataset. The results show that our algorithm is much faster in practice. See \Cref{sec:exp} for details.



\subsection{Technical overview}\label{sec:ov}

\paragraph{Algorithm overview.} 
Our main algorithm is based on a probabilistic recursion.
Consider the linear system $Sz^* = b$, where $S$ is strictly diagonally dominant. 
For any $u\in[n]$, it holds that $\sum_{v\neq u} S_{uv}z_v^*+S_{uu}z_u^*=b_u$. By rearranging the terms, the following equation holds for any $u\in[n]$:
\[z^*_u= \frac{b_u - \sum_{v \neq u} S_{uv} z_v^*}{S_{uu}}.\]

Note that every entry in $S$ can take any sign. Define the following sign function. 
For any $x \in \mathbb{R}$, define the sign function as $\text{sgn}(x) = 1$ if $x > 0$, $\text{sgn}(x) = 0$ if $x = 0$, and $\text{sgn}(x) = -1$ if $x < 0$.
Recall that $\dout_u = \sum_{v \neq u} |S_{uv}|$ is the out-weighted-degree of $u$.
Now, $z^*_u$ can be rewritten as
\begin{align}\label{eq:rec}
    z^*_u= \underbrace{\frac{|S_{uu}| - \dout_u}{|S_{uu}|}}_{=P(u,u)} \cdot \frac{\text{sgn}(S_{uu}) b_u}{|S_{uu}| - \dout_u} + \sum_{v \neq u} \underbrace{\frac{|S_{uv}|}{|S_{uu}|}}_{=P(u,v)}\cdot \sgn(-S_{uu}S_{uv}) z_v^*.
\end{align}

Let $P(u,u) = \frac{|S_{uu}| - \dout_u}{|S_{uu}|}$ for $u \in [n]$ and $P(u,v) = \frac{|S_{uv}|}{|S_{uu}|}$ for $v \neq u$. Since $S$ is strictly diagonally dominant, the matrix $P$ forms a stochastic matrix, which means every row of $P$ is a probability distribution over $[n]$.
Equation~\eqref{eq:rec} can be interpreted as: (1) with probability $P(u,u)$, $z^*_u$ takes the value $\frac{\text{sgn}(S_{uu}) b_u}{|S_{uu}| - \dout_u}$;
(2) for each $v \neq u$, with probability $P(u,v)$, $z^*_u$ takes the value $\sgn(-S_{uu}S_{uv}) z_v^*$.

The above interpretation suggests a simple recursive algorithm. To solve the value of $z^*_u$, the algorithm aims to generate a random variable $z_u$ such that $\E[z_u] = z^*_u$. 
The algorithm first samples $w \in [n]$ according to the distribution $P(u,\cdot)$. If $w = u$, then the recursion terminates and returns value $z_u = \frac{\text{sgn}(S_{uu}) b_u}{|S_{uu}| - \dout_u}$. Otherwise, the algorithm first recursively generates random variable $z_w$ and return $z_u = \sgn(-S_{uu}S_{uv}) z_v^*$. 
The recursion process terminates quickly as $P(v,v) > 0$ for all $v$.
\paragraph{Proof overview of the lower bound.} 
Our lower bound is based on the construction of two distributions of $S \in \mathbb{R}^{n \times n}$ and $b \in \mathbb{R}^n$ instances $\mathcal{F}_0$ and $\mathcal{F}_1$. Suppose one chooses $\mathcal{F} \in \{\mathcal{F}_0,\mathcal{F}_1\}$ and draws a random instance $(S,b)$ from $\mathcal{F}$. Consider the sublinear-time algorithm interacting with the $(S,b)$-oracle. By choosing a proper $u \in [n]$ as the input of the algorithm, we can show the following two propositions:
\begin{itemize}
   \item  If there is a sublinear-time algorithm that approximates $z^*_u$ with a small constant error, where $z^*$ is the unique solution satisfying $Sz^*=b$, then the algorithm can distinguish whether $\mathcal{F}=\mathcal{F}_0$ or $\mathcal{F}=\mathcal{F}_1$.
   \item  Any sublinear-time algorithm that can distinguish whether $\mathcal{F}=\mathcal{F}_0$ or $\mathcal{F}=\mathcal{F}_1$ requires at least $\Omega(S_{\max})$ queries, where $S_{\max} = \max_{ i \in [n]} |S_{ii}|$.
\end{itemize}
Combining the above two propositions, we can show that any sublinear-time algorithm that approximates $z^*_u$ with a small constant error requires at least $\Omega(S_{\max})$ queries.

Next, we sketch the construction of the hard instances. In both distributions $\mathcal{F}_0$ and $\mathcal{F}_1$, the matrix $S = I + L$ is the same, where $I$ is the identity matrix and $L$ is the Laplacian matrix of a weighted graph $G$. 
The graph $G$ is constructed by using constant degree expander graphs as building blocks.
The vector $b$ takes different values in two distributions $\mathcal{F}_0$ and $\mathcal{F}_1$. The two properties mentioned above can be verified by using some expansion properties of the expanders.


\subsection{Other related work}\label{sec:relatedwork}
Random walks have also played a central role in the design of sublinear-time algorithms for various graph problems beyond solving linear systems. They have been used to approximate a variety of local graph centrality measures, such as PageRank and personalized PageRank scores, which are important in ranking and recommendation systems~\cite{bressan2018sublinear}. Random walks are also instrumental in estimating stationary distributions of Markov chains, particularly for fast convergence and mixing time analysis in large networks~\cite{banerjee2015fast,bressan2019approximating}. In addition, random walk techniques have been applied to estimate effective resistances between vertex pairs~\cite{andoni2019solving,peng2021local}, a quantity closely tied to graph connectivity and electrical flow. Another notable application is in designing algorithms for sampling vertices with probability proportional to their degrees, which is useful in streaming and sparsification contexts~\cite{dasgupta2014estimating}. These applications demonstrate the versatility of random walks as a tool for designing efficient sublinear and local algorithms in large-scale graphs.

There are also some quantum-inspired classical sublinear-time algorithms for linear system solvers. Chia et al. gave classical sublinear-time algorithms for solving low-rank linear systems of equations \cite{chia2018quantum,chia2020quantum}. Their algorithms are inspired by the HHL quantum algorithm \cite{harrow2009quantum} for solving linear systems and the breakthrough by Tang \cite{tang2019quantum} of dequantizing the quantum algorithm for recommendation systems. Their sublinear-time algorithms sample from and estimate entries of $x$ such that $\norm{x-A^\dagger b}\leq \eps \norm{A^\dagger b}$ for low-rank $A \in \mathbb{C}^{m \times n}$ in $\poly(k, \kappa, \|A\|_F, 1/\varepsilon) \cdot \poly\log(m, n)$ time per query, assuming length-square sampling access to
input $A$ and query access to $b$.  The length-square access of a vector $v$ allows queries to the norm $\norm{v}$ and
individual coordinates $v_i$
, as well as providing samples from the length-square distribution
$\frac{\abs{v_i}}{\norm{v}^2}$. In these algorithms they assume length-square access to the input matrix $A$, requiring
length-square access to the rows of $A$ as well as to the vector of row-norms. The output of
their algorithms is a description of an approximate solution $\tilde{x}$, providing efficient length-square
access to $\tilde{x}$. Shao and Montanaro \cite{shao2022faster} gave a slightly improved algorithm.

\section{Preliminaries}\label{sec:preliminaries}

Throughout the paper, we use $i\in [n]$ to denote $1\leq i\leq n$. Let $S$ denote a $\delta$-diagonally dominant matrix. We use $\lVert S\rVert_2$ and $\lVert S\rVert_{\infty}$ to denote the $2$-norm and the infinity norm of $S$. The condition number of a matrix \( S \), denoted \( \kappa(S) \), is defined as  
\[
\kappa(S) = \lVert S \rVert \cdot \lVert S^{-1} \rVert.
\]  
The \textbf{2-norm condition number} \( \kappa_2(S) \) and the \textbf{infinity-norm condition number} \( \kappa_{\infty}(S) \) are obtained by replacing the norm with the corresponding norm. In this paper, we abuse the notation to have $S^{-1}$ refer to the pseudo-inverse of $S$ if $S$ is singular. For singular and symmetric $S$, $S^{-1}$ is defined by first doing  diagonalization $S = V\Sigma V^T$ and then let $S^{-1} = V\Sigma^{-1}V^T$, where $\Sigma^{-1}$ takes the reciprocal of each non-zero entry of $\Sigma$.


We often use the coupling argument to bound the total variation distance between two distributions. Let $\nu$ and $\mu$ be two distributions over some set $\Omega$. The total variation distance between $\nu$ and $\mu$ is defined as $d_{\text{TV}}(\nu, \mu) = \sup_{A \subseteq \Omega} |\nu(A) - \mu(A)| = \frac{1}{2} \sum_{x \in \Omega} |\nu(x) - \mu(x)|$. A coupling of $\nu$ and $\mu$ is a joint distribution $(X,Y)$ over $\Omega\times\Omega$ such that the marginal distribution of $X$ is $\nu$ and the marginal distribution of $Y$ is $\mu$. The following coupling inequality is well-known. 
\begin{lemma}[Coupling Inequality \text{\cite[Section 4.1]{levin2017markov}}]\label{lem:coupling}
    Let $\nu$ and $\mu$ be two distributions over some set $\Omega$. For any coupling $(X,Y)$ of $\nu$ and $\mu$, we have
    $$
    d_{\text{TV}}(\nu, \mu) \leq \Pr[X\neq Y].
    $$
In particular, there exists an optimal coupling $(X,Y)$ that achieves the equality.
\end{lemma}

\section{The Algorithms for DD Systems}\label{sec:alg}

In this section, we present the algorithms based on the recursive procedure. 



\paragraph{Basic subroutines}
Using the recursion idea discussed in~\Cref{sec:ov}, Equation~\eqref{eq:rec} suggests the following simple recursive procedure \Cref{alg:recursive} \textsc{RecursiveSolver}. Furthermore, we need the subroutine \Cref{alg:solver} \textsc{EstimateZ} that outputs the average of estimates from some other procedure \textbf{ASolver}.

\begin{algorithm}[H]
\caption{\textsc{RecursiveSolver} ($u$)}
    \label{alg:recursive}
    \begin{algorithmic}[1]
        \STATE \textbf{Input:} Current vertex \( u \); query access to the \((S, b)\)-oracle
        \STATE \textbf{Output:} A real-valued random variable
        \STATE Perform a vertex query on \( u \) to retrieve \( b_u \), \( S_{uu} \), and \( \mathrm{d}^{\mathrm{out}}_u \)
        \STATE \textbf{With probability \( \frac{|S_{uu}| - \mathrm{d}^{\mathrm{out}}_u}{|S_{uu}|} \) do}
        \STATE \hspace{1em} \textbf{return} \( \frac{\mathrm{sgn}(S_{uu}) \cdot b_u}{|S_{uu}| - \mathrm{d}^{\mathrm{out}}_u} \)
        \STATE \textbf{Otherwise}
        \STATE 
        Perform a random walk query on \( u \) to obtain an out-neighbor \( v \) and edge weight \( S_{uv} \) \label{line:s11}
        \STATE 
        \textbf{return} \( \mathrm{sgn}(-S_{uu} S_{uv}) \cdot \textsc{RecursiveSolver}(v) \)
    \end{algorithmic}
\end{algorithm}

\begin{algorithm}[H]  \caption{\textsc{EstimateZ} ($u,\delta,\varepsilon,\lVert b\rVert_\infty$)}
    \label{alg:solver}
    \begin{algorithmic}[1]
        \STATE \textbf{Input:} The query vertex \( u \);  parameters \( \delta > 0 \), \( \varepsilon > 0 \), and \( \|b\|_\infty \); query access to the \((S, b)\)-oracle
        \STATE \textbf{Output:} The estimate of $z^*_u$ 
        \STATE Set \( t \gets \frac{6 \|b\|_\infty^2}{\delta^2 \varepsilon^2} \)
        \FOR{ \( i = 1 \) to \( t \) }
            \STATE \( z_u^{(i)} \gets \textbf{ASolver}(u) \) \label{line:r11}
        \ENDFOR
        \STATE \textbf{return} \( \tilde{z}_u \gets \frac{1}{t} \sum_{i=1}^t z_u^{(i)} \)
    \end{algorithmic}
\end{algorithm}
For any vertex $u$, let $z_u$ denote the random variable returned by \textsc{RecursiveSolver}($u$). We have the following lemma and the proof is in \Cref{sec:recursive}.
\begin{lemma}\label{lem:recursive}
For any $u \in [n]$, \textsc{RecursiveSolver} ($u$) returns a random variable $z_u$ with expected query complexity $O(\frac{S_{\max}}{\delta})$, where $z_u$ satisfies  $\E[z_u] = z^*_u$ and $z^*$ is the unique solution of $Sz=b$.
\end{lemma}

\paragraph{An algorithm with bounded expected running time}
We first give an algorithm with bounded expected query complexity, as given in the  
following lemma. 
\begin{lemma}\label{lem:algorithmexpected}
With probability at least $2/3$, Algorithm~\ref{alg:solver}, where \textbf{ASolver} is instantiated by \textsc{RecursiveSolver} (Algorithm~\ref{alg:recursive}), returns an estimate $\tilde{z}_u$ such that $\abs{\tilde{z}_u - z_u^*} \leq \varepsilon$. The expected query complexity of the algorithm is $O(\frac{\lVert b\rVert_{\infty}^2S_{\max}}{\delta^3\varepsilon^2})$.   
\end{lemma}
\begin{proof}
Note that the algorithm independently calls \textsc{RecursiveSolver}($u$) for $t=\frac{6\lVert b\rVert_{\infty}^2}{\delta^2\varepsilon^2}$ times. Let the corresponding $t$ estimates be $z_u^{(1)},z_u^{(2)},\cdots,z_u^{(t)}$. The output of the algorithm is the average over these estimates, i.e., $\frac{1}{t}\sum_{i}z_u^{(i)}$. \Cref{alg:recursive} can be viewed as a random walk that stops at the current vertex $u$ 
    and outputs the value $\frac{\text{sgn}(S_{uu}) b_u}{|S_{uu}| - \dout_u}$ with probability $\frac{|S_{uu}| - \dout_u}{|S_{uu}|}$ or performs a random walk query on $u$ with probability $\propto|S_{uv}|$. Note that $|S_{ww}| - \dout_w \geq \delta$ for any $w\in V$, and each returned value $z_u^{(i)}$ takes the form $\frac{b_w}{|S_{ww}|-d_w^{out}}$
    or $\frac{-b_w}{|S_{ww}|-d_w^{out}}$ for some $w$. Thus, we have $|z_u^{(i)}| \leq \frac{\Vert b \Vert_\infty}{\delta}$, 
    which implies $\Var[z_u^{(i)}] \leq \E[(z_u^{(i)})^2] \leq \frac{\lVert b\rVert_{\infty}^2}{\delta^2}$. 
    By Chebyshev inequality,  
    $\Pr\left[\left|\frac{1}{t}\sum_{i}z^{(i)}_u-z_u^*\right|\leq\varepsilon\right]\geq1-\frac{\Var[z_u^{(i)}]}{t\varepsilon^2}\geq \frac{5}{6} > \frac{2}{3}.$ 
    Hence, the algorithm has success  probability at least $2/3$.
    By \Cref{lem:recursive}, the expected query complexity of each execution of \textsc{RecursiveSolver}($u$) is $O(\frac{S_{\max}}{\delta})$. This finishes the proof of the lemma by considering the number of iterations. 
\end{proof}

\paragraph{An algorithm with bounded worst-case running time}  
The running time of the algorithm in \Cref{lem:algorithmexpected} is randomized without a fixed upper bound. We now present a modified algorithm with a bounded worst-case running time. Specifically, instead of instantiating \textbf{ASolver}($u$) with \textsc{RecursiveSolver}($u$) in \Cref{alg:solver}, we use \textsc{TruncatedRecursiveSolver}($u,p$), as defined in \Cref{alg:truncated}. Note that \textsc{TruncatedRecursiveSolver}($u$,$p$) maintains a probability value $p$.

\begin{algorithm}[H]
    \caption{\textsc{TruncatedRecursiveSolver}($u$, $p$)}
    \label{alg:truncated}
    \begin{algorithmic}[1]
        \STATE \textbf{Input:} Current vertex $u$; recursion probability $p$; query access to the $(S, b)$-oracle
        \STATE \textbf{Output:} A real-valued random variable
        \STATE Let $t' = \frac{6 \|b\|_\infty^2}{\delta^2 \varepsilon^2}$
        \STATE Perform a vertex query on $u$ to retrieve $b_u$, $S_{uu}$, and $\mathrm{d}^{\mathrm{out}}_u$
        \STATE \textbf{With probability} $\frac{|S_{uu}| - \mathrm{d}^{\mathrm{out}}_u}{|S_{uu}|}$ \textbf{do}\label{line:s12}
        \STATE \hspace{1em} \textbf{return} $\frac{\mathrm{sgn}(S_{uu}) \cdot b_u}{|S_{uu}| - \mathrm{d}^{\mathrm{out}}_u}$\label{line:return-1}
        \STATE Update $p \leftarrow \frac{\mathrm{d}^{\mathrm{out}}_u}{|S_{uu}|}\cdot p$\label{line:update-p}
        \IF{$p \leq \frac{1}{6t'}$}
            \STATE \textbf{return} an arbitrary value (e.g., $0$)\label{line:tc-2}
        \ELSE
            \STATE Perform a random walk query on $u$ to obtain an out-neighbor $v$ and edge weight $S_{uv}$\label{line:r12}
            \STATE \textbf{return} $\mathrm{sgn}(-S_{uu} S_{uv}) \cdot \textsc{TruncatedRecursiveSolver}(v, p)$
        \ENDIF
    \end{algorithmic}
\end{algorithm}


That is, the modified algorithm, denoted by $\mathcal{A}$, is defined as \Cref{alg:solver} \textsc{EstimateZ}, with \textbf{ASolver} instantiated by \Cref{alg:truncated} \textsc{TruncatedRecursiveSolver}$(u, p)$, where the initial value of $p$ is set to $1$. During each recursive call, if \textsc{TruncatedRecursiveSolver}$(u, p)$ does not terminate at Line~\ref{line:return-1}, then $p$ is updated to $p \gets pq$ and passed to the next level of recursion. If, at any point in the recursion, the algorithm finds that $p \leq \frac{1}{6t'}$, where $t' = \frac{6 \|b\|_\infty^2}{\delta^2 \varepsilon^2}$, it terminates and returns an arbitrary value.




We use algorithm $\mathcal{A}$ to prove Theorem~\ref{thm:algo} in \Cref{sec:algo}. We prove \Cref{cor:rel} that achieves relative error with respect to $\lVert z^*\rVert_{\infty}$ in \Cref{sec:rel}.
    For non-strictly diagonally dominant matrix, we prove Theorem~\ref{thm:cond} in \Cref{sec:cond}. As an application of our main theorem, we prove Theorem~\ref{thm:fj} in \Cref{subsec:proofoffj}.

\section{The Lower Bound}\label{sec:lb}
In this section, we prove Theorem \ref{thm:lb}. We start with the construction of hard instances. For infinitely many $n$, we consider the following hard instance $(S,b)$, where $S \in \mathbb{R}^{n \times n}$ and $b \in \mathbb{R}^n$ are both constructed from a weighted graph $G$. We first define the weighted graph $G$ and then show how to construct $S$ and $b$ from $G$.

 
We use unweighted expander graphs as a basic building block of our hard instances.
We first introduce some definitions.
Let graph $G^{\text{nd}}$ be an undirected, unweighted graph. The Laplacian of $G^{\text{nd}}$ is defined as $L = D-A$, where $D$ is the diagonal matrix of degrees and $A$ is the adjacency matrix of $G^{\text{nd}}$. Let $\bar{L} = D^{-1/2}LD^{-1/2}$ be the normalized Laplacian of $G^{\text{nd}}$.
Let $0=\lambda_1\leq\lambda_2\leq\cdots\leq\lambda_n\leq 2$ be the eigenvalues of $\bar{L}$.
The spectral expansion of $G^{\text{nd}}$ is defined as $\gamma_G=\min(\lambda_2,2-\lambda_n)$.

\begin{proposition}[\text{Ramanujan Graphs}~\cite{lubotzky1988ramanujan,hoory2006expander}]\label{prop:expander}
Let $d = 98$, where $d-1$ is a prime number and $(d-1)\mod 4 = 1$. 
For infinitely many $N \in \mathbb{N}$, there exists a connected unweighted $d$-regular simple graph $G^{\textnormal{ex}}_N=(V,E)$ with $N$ vertices such that $\gamma_{G^{\textnormal{ex}}_N}\geq 1 - \frac{2\sqrt{d-1}}{d} >\frac23$.
\end{proposition}

 In the rest of the proof, we use $d$ to denote the universal constant 98 in Proposition \ref{prop:expander}.
Fix $k \in \mathbb{N}$ such that the graph $G^{\textnormal{ex}}_k$ in \Cref{prop:expander} exists. Let
$n = k^2.$ 
We now construct a distribution $\mu_n$ of graphs and then sample a random graph $G \sim \mu_n$ to construct hard instances. By \Cref{prop:expander}, there are infinitely many integers $k$. Hence, we can construct the distribution $\mu_n$ for infinitely many $n$. 
Let graphs $G'$ and $B$ be two disjoint copies of $G^{\textnormal{ex}}_k$.
Let $C$ be a set of $n - 2k$ isolated vertices. 
Note that $G',B$ and $C$ have $n$ vertices in total. 
A random \emph{weighted} graph $G \sim \mu_n$ is constructed by the following procedure:
\begin{itemize}
\item Take the union of graphs $G',B$ and $C$.
\item Sample a vertex $w_{G'}$ in $G'$ and sample a vertex $w_B$ in $B$ uniformly at random. Connect $w_{G'}$ and $w_{B}$ by an unweighted edge; 
\item Let $G$ be the whole graph with $n$ vertices and the labels of vertices are assigned as a uniform random permutation of $[n]$.
\item So far, the graph $G$ is unweighted. For any edge $e$ in $G$, assign the weight $k$ to the edge $e$. 
\item For every vertex $v$ in $G$, independently assign a uniform permutation of all its neighbors, where the permutation is used to answer the neighbor query.
\end{itemize}
See \Cref{fig:lb} for an illustration of the construction.


\begin{figure}[ht]
   \centering
    \includegraphics[width=120mm]{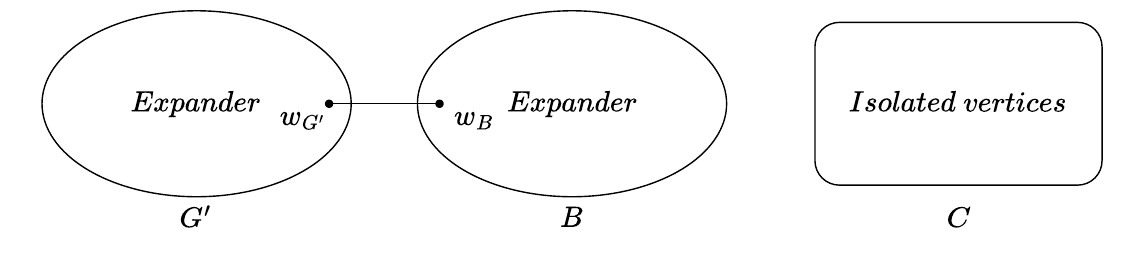}
   \caption{Construction of $G$ with two expanders ($G'$ and $B$) and $n-\Theta(\sqrt{n})$ isolated vertices ($C$). }\label{fig:lb}
\end{figure}

We now sample a random graph $G \sim \mu_n$ and show how to construct two instances $(S^G,b^{G\text{-}1})$ and $(S^G,b^{G\text{-}0})$ from $G$.
Note that two instances share the same matrix $S^G$ but have different vectors $b^{G\text{-}1}$ and $b^{G\text{-}0}$.
 Let $L^G = D^G - A^G$ be the Laplacian of $G$, where $D^G$ is the diagonal matrix of weighted degrees and $A^G$ is the weighted adjacency matrix of $G$. The input matrix is the following:
$S^G = I + L^G$.
The following observation follows directly from the definition.
\begin{observation}\label{obs:S}
The matrix $S^G$ is $1$-DD, $S_{\max} = \Theta(dk) = \Theta(\sqrt{n})$, and $\Delta^S_{\max} = d = O(1)$.
\end{observation}

Given the graph $G$, we define two vectors $b^{G\text{-}1}$ and $b^{G\text{-}0}$ as follows:
\begin{align*}
\forall v \in [n], \quad b^{G\text{-}0}_v = 0, \quad b^{G\text{-}1}_v &= \begin{cases}
        1 & \text{if } v \text{ is a vertex in } B, \\
        0 & \text{otherwise}.
    \end{cases}
\end{align*}
Finally, the query vertex $u$ is sampled uniformly at random from $G' \subseteq G$. As discussed before, we will show that the following two properties hold.
\begin{itemize}
    \item There is a constant gap between $z^1_u$ and $z^0_u$, where $z^1_u$ and $z^0_u$ are the unique solutions of $S^Gz = b^{G\text{-}1}$ and $S^Gz = b^{G\text{-}0}$ respectively. In other words, solving the query at $u$ for two instances with a small constant additive error can distinguish the two instances.
    \item Given the query vertex $u$, any sublinear-time algorithm with low query complexity cannot distinguish between the two instances $(S^G,b^{G\text{-}1})$ and $(S^G,b^{G\text{-}0})$ with high probability.
\end{itemize}

Formally, the first property is stated in \Cref{lem:diff} and the second property is stated in \Cref{lem:tv}. The lower bound result in \Cref{thm:lb} can be proved by combining two lemmas. The detailed analysis is given in~\Cref{sec:lb-pf}.

Finally, if the oracle allows the edge query, i.e., the oracle returns the information of a given edge $e\in E$, we can slightly modify the construction of hard instances to have the same lower bound. See \Cref{remark:edge} for details.

\section{Experiments}\label{sec:exp}

Now we experimentally evaluate our main algorithm, Algorithm~\ref{alg:solver} and instantiate \textbf{Asolver} by \textsc{RecursiveSolver} (\Cref{alg:recursive}). We focus on evaluating the approximation quality and the running complexity of our algorithm on the Friedkin-Johnsen Model. As a baseline, we compare our algorithm with Algorithm $1$ in \cite{neumann2024sublinear}, which we refer to as \textbf{NDP24} (named after the authors). Their algorithm is mainly based on the random walk based on linear system solver in~\cite{andoni2019solving}. We compare two algorithms in two settings: the absolute error varies with the average running time of each vertex and the absolute error with a bounded number of random walks.

\paragraph{Implementation and Dataset.} We run our experiment with Intel(R) Core(TM) i$7$-$9750$H CPU @ 2.60GHZ with $6$ cores and $32$ GB RAM and our code is based on the framework of \cite{neumann2024sublinear}. We implement two algorithms in C++$11$ and all other parts are implemented in Python. We use the GooglePlus, TwitterFollows and Pokec dataset from KONECT~\cite{10.1145/2487788.2488173} and Network
Repository~\cite{nr}. Note that these datasets only consist of an unweighted graph, we generate the innate opinion using a uniform distribution in $[0,1]$. Since our datasets are too large to obtain the exact ground truth, we use the results of the algorithm in~\cite{10.1145/3442381.3449812} based on the Laplacian solver as our ground truth. Their experiments show that the error of their algorithms is negligible in practice, where their algorithm's relative error is less than $10^{-6}$ with the exact results based on matrix inversion.

\paragraph{Absolute error varies with the average running time of each vertex.}  Since the two algorithms have different parameterizations, we first compare how the absolute error varies with the average running time per vertex. We sample $5000$ vertices with replacement and compute the average absolute error and running time for each algorithm. For our Algorithm~\ref{alg:solver}, we set the number of independent samples $t = 100, 1000, 5000, 10000$. For NDP24, we vary the walk length $\ell = 200, 400, 800$ and number of walks $r = 100, 500, 1000, 2000, 5000$. As noted in~\cite{neumann2024sublinear}, due to early termination of walks, increasing the walk length yields diminishing returns in accuracy. This trade-off is evident in Figure~\ref{fig:time}, where our algorithm consistently achieves better approximation quality for the same per-vertex runtime.

\begin{figure}[ht]
    \centering
    \begin{subfigure}{0.31\columnwidth}
        \centering
        \includegraphics[width=\linewidth]{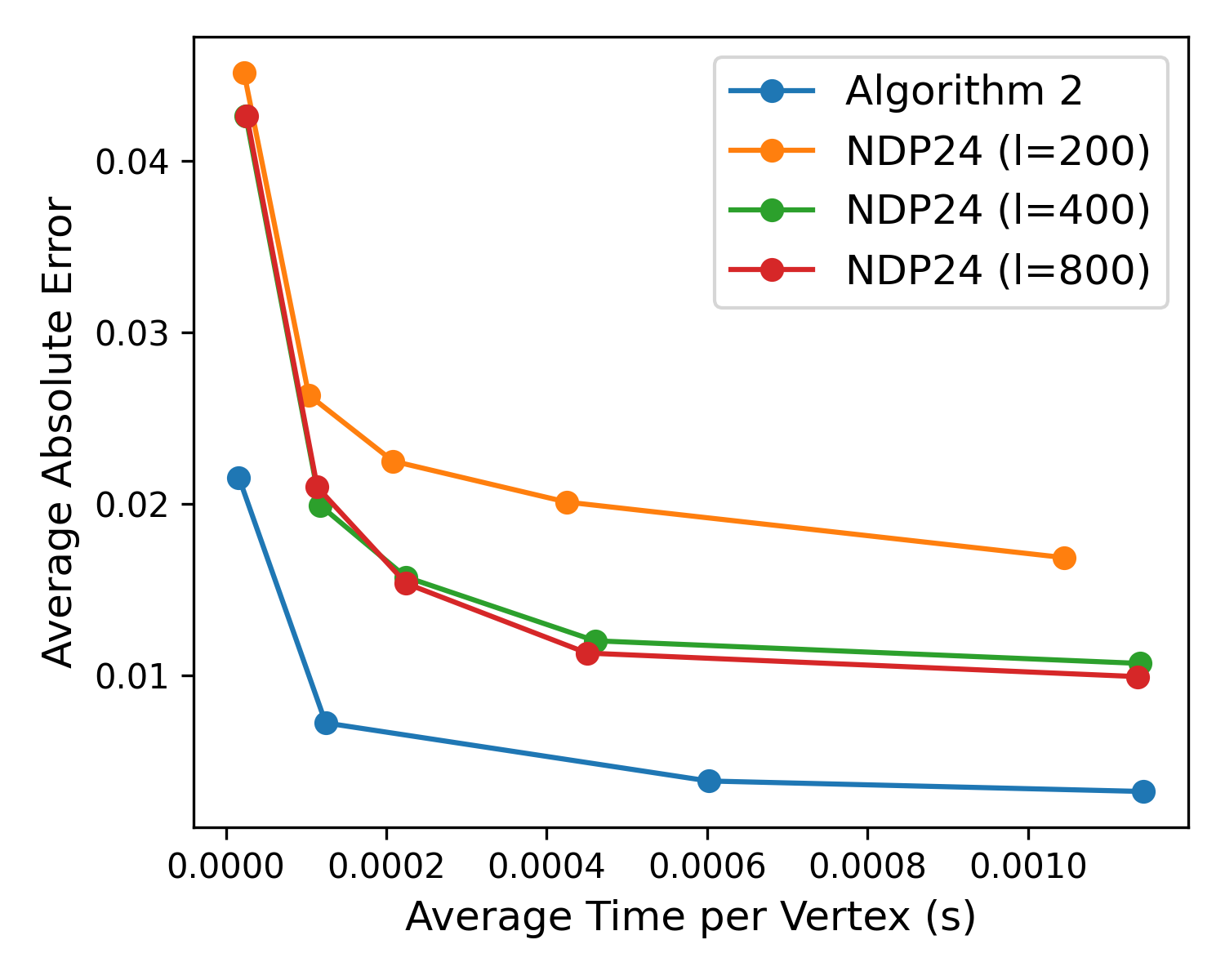}
        \caption{GooglePlus}
        \label{fig:googleplus}
    \end{subfigure}
    \hspace{0.02\columnwidth}%
    \begin{subfigure}{0.31\columnwidth}
        \centering
        \includegraphics[width=\linewidth]{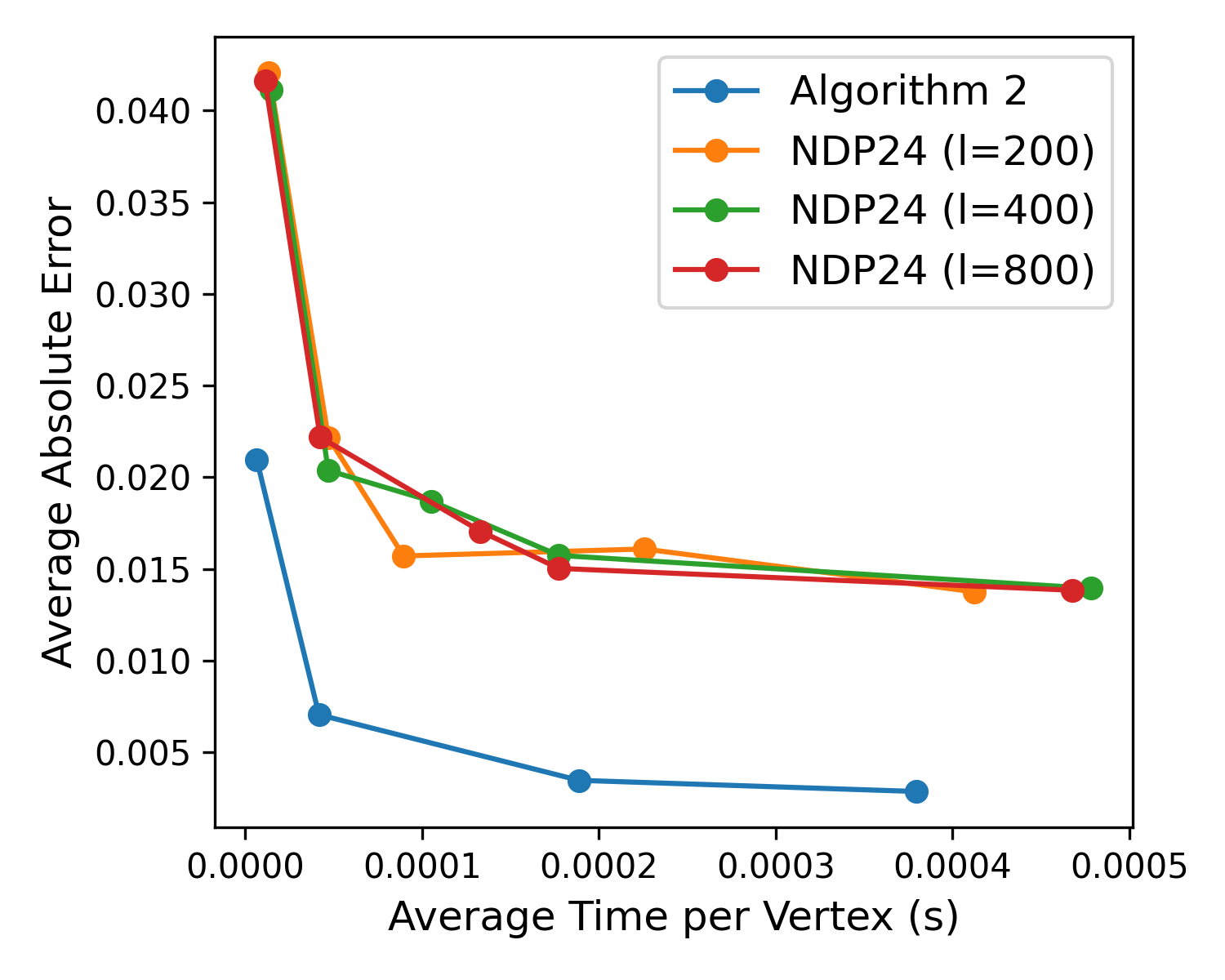}
        \caption{TwitterFollows}
        \label{fig:twitter}
    \end{subfigure}
    \hspace{0.02\columnwidth}%
    \begin{subfigure}{0.31\columnwidth}
        \centering
        \includegraphics[width=\linewidth]{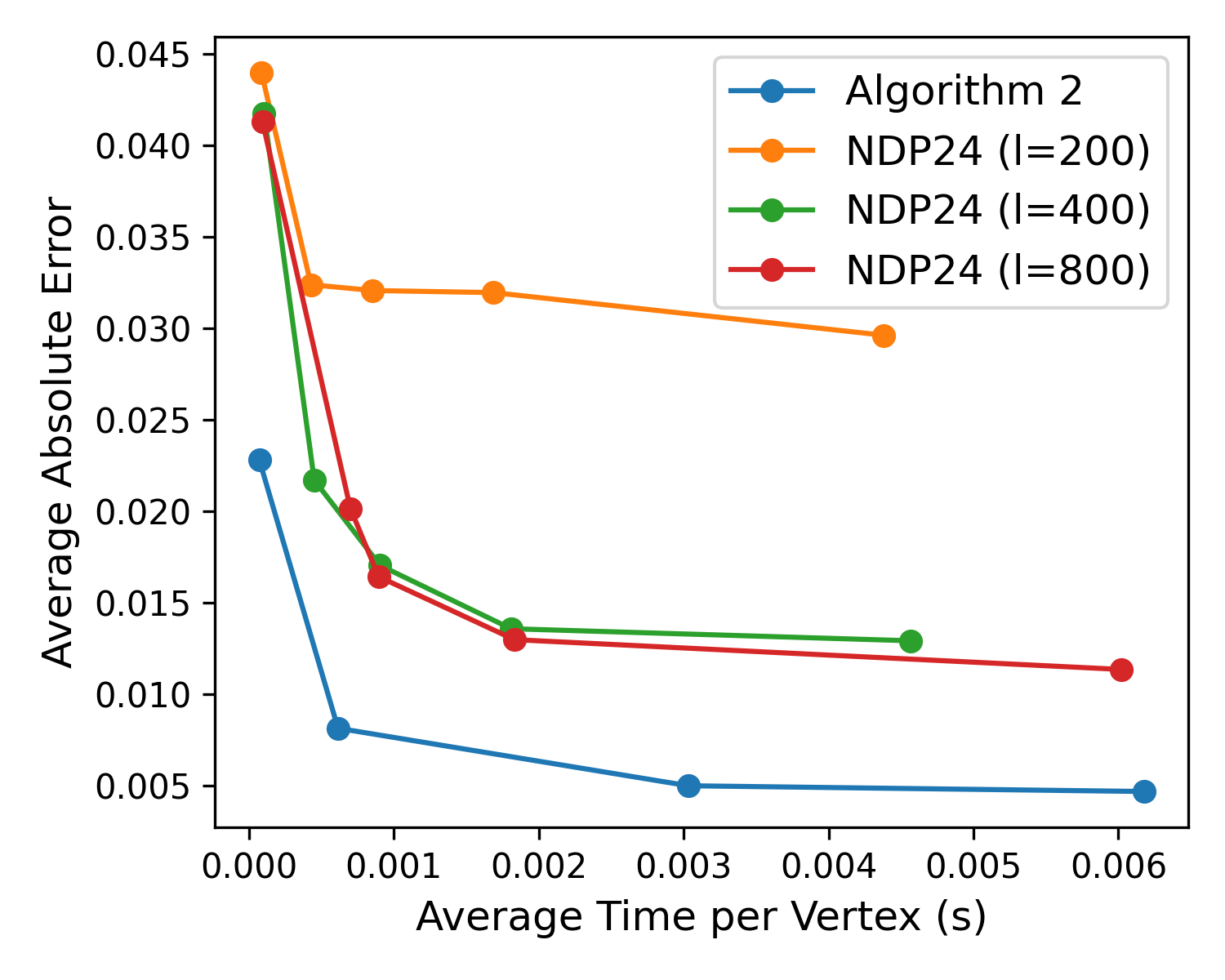}
        \caption{Pokec}
        \label{fig:pokec}
    \end{subfigure}
    \caption{The average absolute error varies with the average running time per vertex of Algorithm~\ref{alg:solver} and NDP24 based on the linear system solver.}
    \label{fig:time}
\end{figure}


\paragraph{Absolute error with a bounded number of random walks.} Both algorithms rely on random walks as their fundamental query type. To compare their performance under a fixed query budget, we analyze the average absolute error using a bounded number of random walk queries. Specifically, we sample $5000$ vertices with replacement and measure the average absolute error for query budgets $q = 5000, 10000, 20000, 40000, 80000$. Each random walk terminates and returns the average result once the query budget is exhausted. For NDP24, we fix the walk length to $\ell = 800$. As shown in Table~\ref{tab:avg_error_networks}, our algorithm consistently achieves lower error under the same query constraints.

\begin{table}[ht]
\centering
\caption{The average absolute error with bounded number of random walk queries}
\label{tab:avg_error_networks}
\begin{tabular}{llrrrrr}
\toprule
\multirow{2}{*}{\textbf{Network}} & \multirow{2}{*}{\textbf{Method}} & \multicolumn{5}{c}{\textbf{Number of Random Walk Queries}} \\
\cmidrule(lr){3-7}
 &  & 5000 & 10000 & 20000 & 40000 & 80000 \\
\midrule
\multirow{2}{*}{GooglePlus} & Algorithm 2 & 0.0100 & 0.0069 & 0.0051 & 0.0038 & 0.0030 \\
                            & NDP24       & 0.0262 & 0.0189 & 0.0154 & 0.0111 & 0.0085 \\
\multirow{2}{*}{Pokec}      & Algorithm 2 & 0.0173 & 0.0124 & 0.0091 & 0.0070 & 0.0054 \\
                            & NDP24       & 0.0438 & 0.0321 & 0.0232 & 0.0187 & 0.0138 \\
\multirow{2}{*}{TwitterFollows} & Algorithm 2 & 0.0067 & 0.0050 & 0.0036 & 0.0028 & 0.0022 \\
                                & NDP24       & 0.0196 & 0.0159 & 0.0164 & 0.0126 & 0.0152 \\
\bottomrule
\end{tabular}
\end{table}



\newpage
\bibliographystyle{plainurl}
\bibliography{ref}

\appendix
\newpage
\section{Omitted Proofs for Algorithms}\label{sec:alg-pf}

\subsection{Proof of Lemma~\ref{lem:recursive}}\label{sec:recursive}

\begin{proof}[\textbf{Proof of \Cref{lem:recursive}}]
    First, we show that the random variable $z_u$ returned by the algorithm is well-defined. In each recursion step, the algorithm terminates and output a value with probability $\frac{|S_{uu}| - \dout_u}{|S_{uu}|} \geq \frac{\delta}{|S_{uu}|} > 0$. Otherwise, the algorithm starts a new recursion step. Hence, the number of recursion steps is dominated by a geometric random variable and the algorithm finally terminates with probability 1. This implies the random variable $z_u$ is well-defined. 
    Furthermore, 
    by the expectation of geometric random variable, the expected recursion depth 
    is at most $O(\max_{i\in[n]}\frac{|S_{ii}|}{|S_{ii}| - \dout_i})\leq O(\frac{S_{\max}}{\delta})$.
    Since the total query complexity is linear in the total number of recursions, the expected query complexity of \textsc{RecursiveSolver} ($u$) is $O(\frac{S_{\max}}{\delta})$.

    Next, we calculate the expectation $\E[z_u]$ of $z_u$. By the definition of the \textsc{RecursiveSolver}($u$), $z_u$ takes the value $\frac{\text{sgn}(S_{uu}) b_u}{|S_{uu}| - \dout_u}$ with probability $\frac{|S_{uu}| - \dout_u}{|S_{uu}|}$. Otherwise, with probability $\frac{|S_{uv}|}{|S_{uu}|}$, $z_u$ takes the random value $\text{sgn}(-S_{uu}S_{uv}) z_v$. By the law of total expectation, we have
    \begin{align*}
        \E[z_u]&= \frac{|S_{uu}| - \dout_u}{|S_{uu}|} \cdot \frac{\text{sgn}(S_{uu}) b_u}{|S_{uu}| - \dout_u} + \sum_{v \neq u} \frac{|S_{uv}|}{|S_{uu}|}\cdot \sgn(-S_{uu}S_{uv}) \E[z_v].\\
        &=S_{uu}^{-1}\left(b_u-\sum_{v\neq u}S_{uv}\E[z_v]\right).
    \end{align*}
    Let $\E[z]$ denote the vector $(\E[z_u])_{u \in [n]}$. Since the above equation holds for any $u \in [n]$, we have $S \E[z] = b$. Since $Sx = b$ has a unique solution $x = z^*$, we have $\E[z] = z^*$. In particular, $\E[z_u] = z^*_u$ for all $u \in [n]$.
\end{proof}

\subsection{Proof of Theorem~\ref{thm:algo}}\label{sec:algo}


\begin{proof}[\textbf{Proof of Theorem~\ref{thm:algo}}]
Recall that $\mathcal{A}$ is defined as \Cref{alg:solver} \textsc{EstimateZ}, with \textbf{ASolver} instantiated by \Cref{alg:truncated} \textsc{TruncatedRecursiveSolver}$(u, p)$, where the initial value of $p$ is set to $1$. 

For any $u \in [n]$, we define two random variables $z^A_u$ and $z^B_u$ as follows:
\begin{align*}
    z^A_u &=  \textsc{RecursiveSolver}(u)\\
    z^B_u &= \textsc{TruncatedRecursiveSolver}(u,1).
\end{align*}
We claim that for all $u \in [n]$, the following total variation distance bound holds:
\begin{align}\label{eq:tv-bound}
    d_{TV}\left(z_u^A,z_u^B\right) \leq \frac{1}{6t},
\end{align}
where $t = \frac{6\lVert b\rVert_{\infty}^2}{\delta^2\varepsilon^2}$. Now, let us assume~\eqref{eq:tv-bound} holds and continue to prove the theorem. The proof of~\eqref{eq:tv-bound} is given at the end of this proof.
Let $\tilde{z}^A_u$ denote the output of \Cref{alg:solver} (with \textbf{ASolver} being instantiated by \Cref{alg:recursive} \textsc{RecursiveSolver}) and $\tilde{z}^B_u$ denote the output of $\mathcal{A}$. By Chebyshev inequality (analysis in the proof of \Cref{lem:algorithmexpected}), it holds that 
\begin{align}\label{eq:chebyshev-2}
    \Pr\left[\left|\tilde{z}^A_u-z_u^*\right|\leq\varepsilon\right]\geq\frac{5}{6}.
\end{align}
By comparing \Cref{alg:solver} with \textbf{ASolver} being instantiated by \Cref{alg:recursive} \textsc{RecursiveSolver} and $\mathcal{A}$. the only difference is that they use different subroutines to sample the value of $z^{(i)}_u$ for $1\leq i \leq t$.
We now construct a coupling between the outputs $\tilde{z}^A_u$ and $\tilde{z}^B_u$ of the two algorithms.
By \eqref{eq:tv-bound}, for each $i \in [t]$, we can use the optimal coupling to successfully couple the random variables $z_u^{(i)}$ returned by \textsc{RecursiveSolver}($u$) and \textsc{TruncatedRecursiveSolver}($u,1$) with probability at least $1 - \frac{1}{6t}$. Hence, by a union bound, $\tilde{z}^A_u$ and $\tilde{z}^B_u$ can be coupled successfully with probability at least $1 - \frac{t}{6t} = \frac{5}{6}$. Therefore, there exists a coupling $\mathcal{D}$ between $\tilde{z}^A_u$ and $\tilde{z}^B_u$ such that
\begin{align}\label{eq:coupling-bound}
    \Pr_\mathcal{D}\left[ \tilde{z}^A_u\neq \tilde{z}^B_u\right] \leq \frac{1}{6}.
\end{align}
By \eqref{eq:coupling-bound} and \eqref{eq:chebyshev-2}, we have
\begin{align*}\label{eq:chebyshev-3}
    \Pr\left[\left|\tilde{z}^B_u-z_u^*\right|\leq\varepsilon\right]\geq\Pr\left[\left|\tilde{z}^A_u-z_u^*\right|\leq\varepsilon\right] - \Pr_{\mathcal{D}}\left[ \tilde{z}^A_u\neq \tilde{z}^B_u\right] \geq \frac{2}{3}.
\end{align*}
Hence, the algorithm $\mathcal{A}$ outputs the correct estimate of $z^*_u$ with probability at least $\frac{2}{3}$.

Next, we analyze the query complexity. In \textsc{TruncatedRecursiveSolver}($u,p$), each recursion multiplies $p$ by a factor of $\frac{\dout_u}{|S_{uu}|}$ at Line \ref{line:update-p}. Since the matrix $S$ is $\delta$-strictly diagonally dominant,
\begin{align*}
    \frac{\dout_u}{|S_{uu}|} \leq \frac{|S_{uu}|-\delta}{|S_{uu}|} = 1-\frac{\delta}{|S_{uu}|} \leq 1-\frac{\delta}{S_{\max}}.
\end{align*}
Let $d = \frac{S_{\max}}{\delta}\log 6t$. After $d$ recursion steps, the value of $p$ is at most
\begin{align*}  
p = \left(1-\frac{\delta}{S_{\max}}\right)^d \leq \exp(-\log 6t) = \frac{1}{6t} = \frac{1}{6t'},
\end{align*}
where $t' = \frac{6 \|b\|_\infty^2}{\delta^2 \varepsilon^2}$ represents the threshold for \textsc{TruncatedRecursiveSolver}($u,1$) for termination. 
Note that per recursion step, the query complexity is $O(1)$.
Hence, the total query complexity of \textsc{TruncatedRecursiveSolver}($u,1$) is at most $O(d)$. Finally, $\mathcal{A}$ calls \textsc{TruncatedRecursiveSolver}($u,1$) $\frac{6\lVert b\rVert_{\infty}^2}{\delta^2\varepsilon^2}$ times, therefore, the total query complexity of $\mathcal{A}$ is
\begin{align*}
    O(td) = O\left(\frac{\lVert b\rVert_{\infty}^2S_{\max}}{\delta^3\varepsilon^2}\log \frac{\lVert b\rVert_{\infty}^2}{\delta^2\varepsilon^2}\right).
\end{align*}



Finally, we verify the total variation distance bound in~\eqref{eq:tv-bound}. In this proof, we only need to consider two algorithms \textsc{RecursiveSolver}($u$) and \textsc{TruncatedRecursiveSolver}($u,1$).
Given the input $u$, we construct a coupling $\mathcal{C}$ between the two algorithms such that
\begin{align*}
    \Pr_{\mathcal{C}}\left[z^A_u\neq z^B_u\right] \leq \frac{1}{6t}.
\end{align*}
Therefore,~\eqref{eq:tv-bound} follows from the coupling inequality.

In Line \ref{line:s11} of \textsc{RecursiveSolver} and Line \ref{line:s12} of \textsc{TruncatedRecursiveSolver}, both algorithms decide whether to terminate the recursion. One can imagine that each algorithm flips a coin with probability $\frac{|S_{uu}| - \mathrm{d}^{\mathrm{out}}_u}{|S_{uu}|}$ of landing on HEADS. If the outcome is HEADS, then the terminates and returns $\frac{\mathrm{sgn}(S_{uu}) \cdot b_u}{|S_{uu}| - \mathrm{d}^{\mathrm{out}}_u}$. We call this step the \emph{coin flip step}.
In Line \ref{line:r11} of \textsc{RecursiveSolver} and Line \ref{line:r12} of \textsc{TruncatedRecursiveSolver}, both algorithms perform a random walk query to obtain a random neighbor of the current vertex. We call this step the \emph{random walk step}.

In our coupling $\mathcal{C}$, initially, both algorithms start at vertex $u$.
We first perfectly couple the coin flip step between the two algorithms. Then, both algorithms may terminate and return the same value. Otherwise, we can perfectly couple the random walk step between the two algorithms. Consequently, in the next recursion, both algorithms move to the same vertex $v$, and we continue this coupling process recursively. It is straightforward to see that in this coupling $\mathcal{C}$, if the outputs $z^A_u$ and $z^B_u$ are different, then it must hold that \textsc{TruncatedRecursiveSolver} terminates at Line \ref{line:tc-2} at some recursion step. In this case, we say that \textsc{TruncatedRecursiveSolver} is truncated.
Formally,
\begin{align*}
    \Pr_{\mathcal{C}}\left[z^A_u \neq z^B_u \right] \leq \Pr_{\mathcal{C}}[\textsc{TruncatedRecursiveSolver} \text{ is truncated}].
\end{align*}
Let $\mathcal{E}$ denote the event that \textsc{TruncatedRecursiveSolver} is truncated. Since $\mathcal{E}$ depends only on the randomness in \textsc{TruncatedRecursiveSolver}, we can focus our analysis on this single algorithm. 
It suffices to prove that 
\begin{align*}
    \Pr[\mathcal{E}] \leq \frac{1}{6t'} = \frac{1}{6t}.
\end{align*}
This inequality is intuitively correct because in Line \ref{line:update-p}, we use $p$ to track the probability that the algorithm does not terminate. The algorithm is truncated when $p \leq \frac{1}{6t'}$, and this event occurs with probability at most $\frac{1}{6t'}$.

We now present a rigorous proof of the above inequality.
Based on our previous analysis, we can decompose the randomness in the algorithm  into two components: $\mathcal{R}_1$, the randomness for all coin flip steps and $\mathcal{R}_2$, the randomness for all random walk steps. Let us fix the randomness $\mathcal{R}_2 = r_2$. Our goal is to prove that
\begin{align*}
    \Pr_{\mathcal{R}_1}[\mathcal{E} \mid \mathcal{R}_2 = r_2] \leq \frac{1}{6t'} = \frac{1}{6t}.
\end{align*}
By averaging over the randomness of $\mathcal{R}_1$, we obtain $\Pr[\mathcal{E}] \leq \frac{1}{6t}$, which implies~\eqref{eq:tv-bound}.
Since we have fixed the randomness of $\mathcal{R}_2$, in every random walk step, the algorithm moves deterministically to the next vertex. Therefore, the randomness $\mathcal{R}_2$ determines an infinite\footnote{The algorithm may terminate within finite number of steps. Hence, it only uses a prefix of the path. In the analysis, we can first generate an infinite path in advance and simulate the algorithm on the path.} long path $(v_i)_{i=0}^\infty$, where $v_0 = u$.  For each vertex $v_i$, if the algorithm does not terminate at $v_i$, it multiplies $p$ by a factor of $\frac{d_{v_i}^{out}}{|S_{v_i v_i}|}$.
We say the algorithm survives until step $i$ if it does not terminate at any step $j < i$.
Let $k$ be the smallest index such that the following condition holds:
\begin{align*}
    \prod_{j=0}^k \frac{d_{v_j}^{out}}{|S_{v_j v_j}|} \leq \frac{1}{6t'}.
\end{align*}
If the event $\mathcal{E}$ occurs, then the algorithm must survive until step $k$, at which point the probability value $p$ maintained by the algorithm becomes $\prod_{j=0}^k \frac{d_{v_j}^{out}}{|S_{v_j v_j}|} < \frac{1}{6t'}$, causing the algorithm to truncate. Conversely, by the randomness in $\mathcal{R}_1$, the algorithm survives until step $k$ with probability $\prod_{j=0}^k \frac{d_{v_j}^{out}}{|S_{v_j v_j}|} \leq \frac{1}{6t'}$. Therefore, we conclude that
\begin{align*}
    \Pr_{\mathcal{R}_1}[\mathcal{E} \mid \mathcal{R}_2 = r_2] \leq \frac{1}{6t'} = \frac{1}{6t}.
\end{align*}
This completes the proof.
\end{proof}

\subsection{Proof of Corollary~\ref{cor:rel}}\label{sec:rel}

\begin{proof}[\textbf{Proof of Corollary~\ref{cor:rel}}]
    By the property of the strictly diagonally dominant matrix, let $z^*$ be any solution to $Sz=b$ and we have
    \[
        \frac{|b_i|}{\left|S_{ii}\right|}=\frac{\left|\sum_{j\in[n]}S_{ij}z^*_j\right|}{\left|S_{ii}\right|}\leq\frac{\sum_{j\in[n]}\left|S_{ij}\right|\lVert z^*\rVert_{\infty}}{\left|S_{ii}\right|}\leq2\lVert z^*\rVert_{\infty}.
    \] 
Therefore, it holds that $\lVert D^{-1}b\rVert_{\infty}\leq 2\lVert z\rVert_{\infty}$, where $D=\text{diag}(S_{11},\cdots,S_{nn})$. Setting $\varepsilon=\frac{\varepsilon\lVert b\rVert_{\infty}}{2\cdot S_{\max}}$ in Theorem~\ref{thm:algo}, then the corresponding algorithm outputs $z_u$ such that $\abs{z_u-z_u^*}\leq \frac{\eps\norm{b}_\infty}{2\cdot S_{\max}}\leq \frac{\eps\norm{b}_\infty}{2\cdot\lVert b\rVert_{\infty}/\lVert D^{-1}b\rVert_{\infty}}=\frac{\eps}{2} \cdot \lVert D^{-1}b\rVert_{\infty}\leq \eps \norm{z^*}_\infty$, where the second inequality follows from the fact that $\frac{\lVert b\rVert_{\infty}}{\lVert D^{-1}b\rVert_{\infty}}\leq S_{\max}$. The query complexity is then $O(\frac{\lVert b\rVert_{\infty}^2S_{\max}}{\delta^3(\varepsilon\lVert b\rVert_{\infty}/S_{\max})^2}\log \frac{\lVert b\rVert_{\infty}}{\delta(\varepsilon\lVert b\rVert_{\infty}/S_{\max})})=O\left( \frac{S_{\max}^3}{\delta^3\varepsilon^2}\log \frac{S_{\max}}{\delta\varepsilon}\right)$.
%
\end{proof}

\subsection{Proof of Theorem~\ref{thm:cond}}\label{sec:cond}

In this section, we prove Theorem~\ref{thm:cond}.
The main idea of the proof is given as follows.
\begin{itemize}
    \item We define a diagonal matrix $I'$ and a small parameter $\sigma$ (detailed definitions are given later). We show that the solution of $Sz =b$ and the solution of $S'z =b$ are close, where $S' = S+\sigma I'$.
    \item Now, $S'$ is a strictly diagonally dominant matrix. We can use the algorithm in Corollary~\ref{cor:rel} to solve the system $S'z =b$.
\end{itemize}

Let $S$ be the input matrix in \Cref{thm:cond}, which satisfies one of the conditions in \Cref{cond:S}.
Defined the following matrix
\begin{align}\label{eq:def-I'}
I'_{ij}&=
\begin{cases}
    \text{sgn}(S_{ii}) &\text{if } i=j \land S_{ii} \neq 0\\
    1 &\text{if } i=j \land S_{ii} = 0\\
    0 &\text{otherwise.}
\end{cases}
\end{align}
Define parameter $\sigma$ and matrix $S'$ by
\begin{align}\label{eq:defSsigma}
    \sigma=\frac{S_{\max}}{\left(\frac{2}{\varepsilon}+1\right)\kappa_{\infty}(S)}, \quad S'=S+\sigma I'.  
\end{align}

By the definition of $I'$, for any $i \in [n]$, $|S'_{ii}| = |S_{ii}| + \sigma$. The following observation is useful.
\begin{observation}
The matrix $S'$ is strictly diagonally dominant. Hence, $S'$ is also non-singular.
\end{observation}

In the second case of \Cref{cond:S}, $S$ is non-singular. We have the following lemma.

\begin{lemma}\label{lem:non-singular}
Suppose $S$ is non-singular. Let $z^* = S^{-1}b$ and $\tilde z^* = (S')^{-1} b$. Then 
\begin{align*}
    \lVert \tilde z^*-z^*\rVert_{\infty}\leq \frac{\varepsilon}{2} \lVert z^*\rVert_{\infty} \quad \text{and} \quad \lVert \tilde z^*\rVert_{\infty}\leq(\frac{\varepsilon}{2}+1)\lVert z^*\rVert_{\infty}. 
\end{align*}
\end{lemma}

In the second case of \Cref{cond:S}, all the diagonal entries are non-zero and have the same sign, say the sign is $c \in \{-1,1\}$.  We can decompose $S=V \Sigma V^T$ using diagonalization. We abuse the notation to use $S^{-1}$ to refer to the pseudo-inverse of $S$ if $S$ is singular, where the pseudo-inverse is defined by $V\Sigma^{-1} V^T$ and $\Sigma^{-1}$ only takes the reciprocal for non-zero diagonal entries.
\begin{lemma}\label{lem:symmetric}
Suppose $S$ is symmetric, and all the diagonal entries are non-zero and have the same sign, say the sign is $c \in \{-1,1\}$. Let $z^* = S^{-1}b$ and $\tilde z^* = (S')^{-1} b$. Then 
\begin{align*}
    \lVert \tilde z^*-z^*\rVert_{\infty}\leq \frac{\varepsilon}{2} \lVert z^*\rVert_{\infty}\quad \text{and} \quad \lVert \tilde z^*\rVert_{\infty}\leq(\frac{\varepsilon}{2}+1)\lVert z^*\rVert_{\infty}.
\end{align*}
\end{lemma}

We now prove the theorem assuming \Cref{lem:non-singular} and \Cref{lem:symmetric}.
\begin{proof}[Proof of Theorem~\ref{thm:cond}]
    We show that we can construct a $(S+\sigma I',b)$-oracle from the $(S,b)$-oracle.
    Note that $\sigma$ can be computed from the input.
    Note that we only add $\sigma I'$ to the diagonal of $S$, this does not influence the neighbor query and random walk query. We can construct the $(S+\sigma I',b)$-oracle by changing the $S_{uu}$ returned by the vertex query on $u$ to $S_{uu}+\sigma \text{sgn}(S_{uu})$ if $S_{uu} \neq 0$; or changing $S_{uu}$ to $S_{uu}+\sigma$ if $S_{uu}=0$. 
    We apply the same algorithm in Corollary \ref{cor:rel} with error bound $\frac{\varepsilon}{10}$ and get $\tilde z_u$ such that $|\tilde z_u-\tilde z_u^*|\leq \frac{\varepsilon}{10}\lVert \tilde z^*\rVert_{\infty}$. With \Cref{lem:non-singular} and \Cref{lem:symmetric}, we have
    \[
        |\tilde z_u-z_u^*|\leq|\tilde z_u-\tilde z_u^*|+|\tilde z_u^*-z_u^*|\leq \frac{\varepsilon}{10}\lVert \tilde z^*\rVert_{\infty}+\frac{\varepsilon}{2}\lVert z^*\rVert_{\infty} \overset{(\star)}{\leq} \left(\frac{\varepsilon}{10}\left(\frac{\varepsilon}{2}+1\right)+\frac{\varepsilon}{2}\right)\lVert z^*\rVert_{\infty} <\varepsilon \lVert z^*\rVert_{\infty}.
    \]
    The last inequality follows from $\lVert \tilde z^*\rVert_{\infty}\leq(\frac{\varepsilon}{2}+1)\lVert z^*\rVert_{\infty}$ in \Cref{lem:non-singular} and \Cref{lem:symmetric}. 
    Note that $S'$ is $\delta$-diagonally dominant with $\delta = \sigma$.
    By \Cref{cor:rel}, the query complexity is 
    $$O\left(\frac{S^3_{\max}}{\sigma^3\varepsilon^2}\log\frac{S_{\max}}{\varepsilon\sigma}\right)=O\left(\frac{\kappa_{\infty}(S)^3}{\varepsilon^5}\log \frac{\kappa_\infty(S)}{\varepsilon}\right).$$
    This completes the proof.
\end{proof}

We now prove \Cref{lem:non-singular} and \Cref{lem:symmetric}. We give a basic property of $\sigma$, which are used in the proof of two lemmas. Note that $\sigma=\frac{S_{\max}}{\left(\frac{2}{\varepsilon}+1\right)\kappa_{\infty}(S)}$. By definition, $\kappa_{\infty}(S) = \Vert S \Vert_\infty \Vert S^{-1} \Vert_\infty$. Note that $\Vert S \Vert_\infty \geq S_{\max}$.
Hence, we have
\begin{align}\label{eq:sigmabd}
 \sigma\leq \frac{1}{\left(\frac{2}{\varepsilon}+1\right)\lVert S^{-1}\rVert_{\infty}}. 
\end{align}

\begin{proof}[Proof of \Cref{lem:non-singular}]
    Let $\tilde z^*$ denote the unique solution satisfying $S'\tilde z^*=b$. We first prove that $|\tilde z^*-z^*|\leq\frac{\varepsilon}{2}\lVert z^*\rVert_{\infty}$, where $z^* = S^{-1}b$. By the definition of $S'$, it holds that $(S+\sigma I')\tilde z^* =b$. By left-multiplying both sides with $S^{-1}$, we obtain the following equation
    \begin{align}\label{eq:pff}
        (I+\sigma S^{-1}I')\tilde z^* = S^{-1} b = z^*.
    \end{align}
    We now show that $(I+\sigma S^{-1}I')$ is non-singular. Suppose $(I+\sigma S^{-1}I')$ is singular. Then there exists a non-zero vector $x$ such that $(I+\sigma S^{-1}I')x = 0$, which implies $x = -\sigma S^{-1}I'x$. Taking infinity norm on both sides implies $\Vert x\Vert_\infty = \Vert \sigma S^{-1}I'x \Vert_\infty$. Hence, the induced infinity norm $\Vert \sigma S^{-1}I' \Vert_\infty \geq 1$. On the other hand, \eqref{eq:sigmabd} implies the following bound
\begin{align}\label{eq:infbd}
 \lVert S^{-1}\sigma I'\rVert_{\infty} \leq \sigma \lVert S^{-1}\rVert_{\infty}\lVert I'\rVert_{\infty} \leq \frac{1}{\frac{2}{\varepsilon}+1} < 1.     
\end{align}
This gives a contradiction. Hence, $(I+\sigma S^{-1}I')$ is non-singular.

The inverse of $(I+\sigma S^{-1}I')$ exists.
We multiply $(I+\sigma S^{-1}I')^{-1}$ on both sides of~\eqref{eq:pff}. This gives 
\begin{align}\label{eq:a1}
    \tilde z^* = (I+\sigma S^{-1}I')^{-1} z^*.
\end{align}

We use the following basic claim in the linear algebra.
\begin{claim}\label{claim:1}
Let $B\in \mathbb{R}^{n \times n}$ be matrix with $\Vert B\Vert_\infty \leq 1$. Suppose $I-B$ is non-singular. Then $\Vert (I - B)^{-1} \Vert_\infty \leq \frac{1}{1-\Vert B \Vert_\infty}$.
\end{claim}
The proof of the claim is deferred to the end of \Cref{sec:cond}.
We now use the claim with $B = -S^{-1}\sigma I'$ and~\eqref{eq:infbd} to obtain
\begin{align*}
    \Vert (I+\sigma S^{-1}I')^{-1} \Vert_\infty \leq \frac{1}{1 - \lVert - S^{-1}\sigma I'\rVert_{\infty}} =  \frac{1}{1 - \lVert S^{-1}\sigma I'\rVert_{\infty}}\leq \frac{1}{1 - \frac{1}{\frac{2}{\varepsilon}+1}} = 1 + \frac{\varepsilon}{2}.
\end{align*}
Combining the above bound with~\eqref{eq:a1}, it holds that
\begin{align}\label{eq:a2}
   \Vert \tilde z^* \Vert_\infty = \Vert (I+\sigma S^{-1}I')^{-1} z^* \Vert_\infty \leq \left( 1 + \frac{\varepsilon}{2} \right) \Vert  z^* \Vert_\infty.
\end{align}

By definition if $\tilde z^*$, we have $(S + \sigma I')\tilde z^* = b$. We can further write it as $(S + \sigma I')(z^* + (\tilde z^*-z^*)) = b$. Expanding it gives $Sz^* + S(\tilde z^*-z^*) + \sigma I' \tilde z^* = b$. Since $S z^* = b$, rearranging gives 
\begin{align}\label{eq:a3}
    S(\tilde z^*-z^*) + \sigma I' \tilde z^* = 0. \quad \Leftrightarrow \quad \tilde z^*-z^* = - S^{-1}\sigma I' \tilde z^*.
\end{align}
Taking the infinity norm on both sides, we have 
\begin{align*}
  \lVert \tilde z^*-z^*\rVert_{\infty}\leq \Vert S^{-1}\sigma I' \Vert_\infty \lVert \tilde z^*\rVert_{\infty}\leq \frac{1}{\frac{2}{\varepsilon}+1}  \Vert \tilde z^*\rVert_{\infty},
\end{align*}
where the last inequality follows from~\eqref{eq:infbd}. Using~\eqref{eq:a2} to bound $\Vert \tilde z^*\rVert_{\infty}$, the following bound holds
\begin{align}\label{eq:a4}
  \lVert \tilde z^*-z^*\rVert_{\infty} \leq \frac{1 + \frac{\varepsilon}{2}}{\frac{2}{\varepsilon}+1} \Vert z^* \Vert_\infty = \frac{\varepsilon}{2}  \Vert z^* \Vert_\infty.
\end{align}
The lemma is proved by combining~\eqref{eq:a4} and~\eqref{eq:a2}.
\end{proof}

\begin{proof}[Proof of \Cref{lem:symmetric}]
Now we assume that  $S$ is symmetric, and all the diagonal entries are non-zero and have the same sign $c \in \{-1,1\}$.
Note that $I' = cI$ in this case, where $I$ is the identity matrix.



Recall that $S'=S+\sigma I'$ and $\tilde z^*$ is the unique solution of $S'z^*=b$. Note that $\tilde z^*$ exists because $S'$ is strictly diagonally dominant and hence it is non-singular.
Similarly, we prove that $|\tilde z^*-z^*|\leq\frac{\varepsilon}{2}\lVert z^*\rVert_{\infty}$, where $z^* = S^{-1}b$.  By definition of $S'$, $(S+\sigma I')\tilde z^*=b$. Recall that $S=V\Sigma V^T$ by diagonalization. By left-multiplying both sides with $S^{-1} = V\Sigma^{-1}V^T$, where $\Sigma^{-1}$ only takes reciprocals on none-zero terms, we obtain
\begin{align}\label{eq:c1}
(\Pi+S^{-1}\sigma I')\tilde z^*=z^*.    
\end{align}

Here $\Pi=V\tilde I V^T$, where $\tilde I$ is defined as follows:
\begin{align*}
\tilde I_{ij}&=
\begin{cases}
    1 &\text{if } i=j\text{ and }\Sigma_{ii}\neq 0\\
    0 &\text{otherwise}
\end{cases}.
\end{align*}
Note that $I' = cI$. 
Then $(\Pi+S^{-1}\sigma I') =V (\tilde I + \sigma c\Sigma^{-1}) V^T$. Taking pseudo-inverse implies $(\Pi+S^{-1}\sigma I')^{-1} =V (\tilde I + \sigma c\Sigma^{-1})^{-1} V^T$. Note that $\tilde I_{ii} \neq 0$ if and only if $\Sigma_{ii} \neq 0$. Left-multiplying $(\Pi+S^{-1}\sigma I')^{-1}$ on both sides of~\eqref{eq:c1} implies 
\begin{align*}
    \Pi \tilde z^* = (\Pi+S^{-1}\sigma I')^{-1} z^*.
\end{align*}
Taking infinity norm on both sides
\begin{align}\label{eq:c2}
    \lVert \Pi \tilde z^*\rVert_{\infty}\leq\lVert (\Pi+S^{-1}\sigma I')^{-1}\rVert_{\infty}\lVert z^*\rVert_{\infty}.
\end{align}
The following claim can be obtained from linear algebra.
\begin{claim}\label{claim:2}
$\lVert (\Pi+S^{-1}\sigma I')^{-1}\rVert_{\infty}\leq\frac{1}{1-\lVert S^{-1}\sigma I'\rVert_{\infty}}$.
\end{claim}
The proof of the claim is deferred to the end of \Cref{sec:cond}.
We now assume the claim is true.
To use the claim on $\lVert (\Pi+S^{-1}\sigma I')^{-1}\rVert_{\infty}$, we need to bound the infinity norm of $\lVert S^{-1}\sigma I'\rVert_{\infty}$. Note that $I' = cI$, where $c\in\{1,-1\}$. We can write 
\begin{align}\label{eq:infbd2}
    \lVert S^{-1}\sigma I'\rVert_{\infty} \leq \sigma \Vert S^{-1}\Vert_\infty \leq \frac{1}{\frac{2}{\varepsilon}+1},
\end{align}
where the last inequality follows from~\eqref{eq:sigmabd}.
Using \Cref{claim:2}, $\lVert (\Pi+S^{-1}\sigma I')^{-1}\rVert_{\infty}\leq\frac{1}{1-\lVert S^{-1}\sigma I'\rVert_{\infty}}\leq \frac{1}{1 - \frac{1}{\frac{2}{\varepsilon}+1}} = \frac{\varepsilon}{2}+1$. Combining the infinite norm bound with~\eqref{eq:c2}, we have
\begin{align}\label{eq:c3}
    \lVert \Pi \tilde z^*\rVert_{\infty}\leq\left(\frac{\varepsilon}{2}+1\right)\lVert z^*\rVert_{\infty}.
\end{align}

Now, we compare two vectors  $\tilde z^*$ and $z^*$. Note that $S' = S+\sigma I' = S + c\sigma I$. Two matrices $S$ and $S'$ share the same orthogonal eigenvectors $v_1,v_2,\cdots,v_n$, which are column vectors of $V$. Let $\lambda_1,\ldots,\lambda_n$ be the eigenvalues of $S$ and $\lambda_1+c\sigma,\ldots,\lambda_n+c\sigma$ be the eigenvalues of $S'$. By decomposing $\tilde z^*$ using the eigenvectors, 
\begin{align*}
    S' \tilde{z}^* = \sum_{i=1}^n (\lambda_i+c\sigma) \langle \tilde{z}^*, v_i \rangle v_i = b
\end{align*}
By comparing the eigenvalues of $S$ and $S'$, since $Sz^* = b$, we have $(\lambda_i+c\sigma) \langle \tilde{z}^*, v_i \rangle=\lambda_i\langle {z}^*,  v_i \rangle $. 
For any $i\in [n]$, if $\lambda_i = 0$, then $(\lambda_i+c\sigma) \langle \tilde{z}^*, v_i \rangle=0$. Since $c \sigma > 0$, it must hold $\langle \tilde{z}^*, v_i \rangle = 0$.
Recall that $\Pi = V \tilde I V^T$. For any $i \in [n]$, $\tilde I_{ii} = 1$ if $\lambda_i \neq 0$ and $\tilde{I}_{ii} = 0$ if $\lambda_i = 0$. 
This implies $\Pi \tilde{z}^* = \tilde{z}^*$. Therefore, by~\eqref{eq:c3}, we have
\begin{align}\label{eq:b1}
    \lVert \tilde z^*\rVert_{\infty}=\lVert \Pi\tilde z^*\rVert_{\infty}\leq\left(\frac{\varepsilon}{2}+1\right)\lVert z^*\rVert_{\infty}.
\end{align}
Recall that $Sz^*=b$ and $(S+\sigma I')(z^*+(\tilde z^*-z^*))=b$, similar to \eqref{eq:a3} we have
\begin{align*}
    S(\tilde z^*-z^*) + \sigma I' \tilde z^* = 0.
\end{align*}
We multiplying $S^{-1}$ on both sides implies 
\begin{align*}
    \Pi(\tilde z^*-z^*) = -S^{-1}\sigma I' \tilde z^*.
\end{align*}
We have verified $\Pi \tilde z^*  = \tilde z^*$. We now verify $\Pi z^* = z^*$. Note that $z^*=S^{-1}b$. By our definition of $S^{-1}$, $\langle {z}^*, v_i \rangle = 0$ for all $i \in [n]$ with $\lambda_i =0$. Similarly, by the definition of $\Pi$, we can verify that $\Pi z^* = z^*$. Next, we have 
\begin{align}
    \lVert \Pi\tilde z^*-\Pi z^*\rVert_{\infty}=\lVert \tilde z^*-z^*\rVert_{\infty}\leq\lVert S^{-1} \sigma I'\rVert_{\infty}\lVert \tilde z^*\rVert_{\infty}\leq\frac{1}{\frac{2}{\varepsilon}+1}\lVert\tilde z^*\rVert_{\infty}, \label{eq:b2}
\end{align}
where the last inequality follows from~\eqref{eq:infbd2}.
Combine (\ref{eq:b1}) and (\ref{eq:b2}), we have
\begin{align}
    \lVert \tilde z^*-z^*\rVert_{\infty}\leq \frac{1 + \frac{\varepsilon}{2}}{\frac{2}{\varepsilon}+1} \lVert z^*\rVert_{\infty} = \frac{\varepsilon}{2}\lVert z^*\rVert_{\infty}.\label{eq3}
\end{align}
The lemma is proved by combining~\eqref{eq:b1} and~\eqref{eq:b2}.
\end{proof}


Now we prove \Cref{claim:1} and \Cref{claim:2}.

\begin{proof}[Proof of \Cref{claim:1}]
    We first prove \Cref{claim:1}. Note that we have the following expression
    \[
        \left(\sum_{k=0}^N B^k\right)(I-B)=I-B^{N+1}.
    \]
    Since $\lVert B\rVert_{\infty}<1$, we have $\lim_{k\rightarrow \infty}B^{k}=0$ because $\lVert B^k\rVert_{\infty}\leq\lVert B\rVert_{\infty}^k$. Then we get
    \[
        \left(\lim_{N\rightarrow\infty}\sum_{k=0}^N B^k\right)(I-B)=I.
    \]
    It follows that $(I-B)^{-1}=\lim_{N\rightarrow\infty}\sum_{k=0}^N B^k$. By taking infinity norm on both sides, we get
    \begin{align*}
        \lVert(I-B)^{-1}\rVert_{\infty}\leq\sum_{k=0}^{\infty}\lVert B^k\rVert_{\infty}\leq\sum_{k=0}^{\infty}\lVert B\rVert_{\infty}^k=\frac{1}{1-\lVert B\rVert_{\infty}}. &\qedhere
    \end{align*}
\end{proof}

The proof of \Cref{claim:2} is similar, but we replace $I$ by $\Pi$ defined above.

\begin{proof}[Proof of \Cref{claim:2}]
    Recall that $S=V\Sigma V^{T}$, $S^{-1}\sigma I'=V\sigma c\Sigma^{-1}V^T$ and $\Pi=V\tilde IV^{T}$, let $B=-S^{-1}\sigma I'$, we have $B\Pi=B$ because $\tilde I_{ii} = 0$ if and only if $\Sigma_{ii} = 0$. Thus, we have the similar expression
    \[
        \left(\sum_{k=0}^N B^k\right)(\Pi-B)=\Pi-B^{N+1}.
    \]
    Since $\Vert B \Vert_\infty  = \lVert S^{-1}\sigma I'\rVert_{\infty}\leq \sigma \lVert S^{-1}\rVert_{\infty}\lVert I'\rVert_{\infty}< 1$, we have $\lim_{k\rightarrow \infty}B^{k}=0$ because $\lVert B^k\rVert_{\infty}\leq\lVert B\rVert_{\infty}^k$. Then we get 
    \[
        \left(\lim_{N\rightarrow\infty}\sum_{k=0}^N B^k\right)(\Pi-B)=\Pi.
    \]
    Note that the pseudo-inverse of $(\Pi-B)$ is $(\Pi-B)^{-1}=V (\tilde I + \sigma c\Sigma^{-1})^{-1} V^T$. We right-multiply $(\Pi-B)^{-1}$ on both sides of the above equation. Then $(\Pi - B)(\Pi-B)^{-1} = \Pi$. Note that $B \Pi = B$. We have
    \[
        \lim_{N\rightarrow\infty}\sum_{k=0}^N B^k=\Pi(\Pi-B)^{-1}=(\Pi-B)^{-1},
    \]
    where in the last equation, we use the fact that $\tilde I_{ii} = 0$ if and only if $\Sigma_{ii} = 0$.
    By taking infinity norm on both sides, we have
     \begin{align*}
        \lVert(\Pi-B)^{-1}\rVert_{\infty}\leq\sum_{k=0}^{\infty}\lVert B^k\rVert_{\infty}\leq\sum_{k=0}^{\infty}\lVert B\rVert_{\infty}^k=\frac{1}{1-\lVert B\rVert_{\infty}}. &\qedhere
    \end{align*}
\end{proof}

\subsection{Proof of Theorem~\ref{thm:fj}}
\label{subsec:proofoffj}
\begin{proof}[\textbf{Proof of Theorem~\ref{thm:fj}}]
By the definition of Friedkin-Johnsen model, we have $S=I+L$ and $b\in[0,1]^n$, where $L$ is the Laplacian matrix of an undirected weighted graph. We can set $\delta=1$, $\lVert b\rVert_{\infty}\leq1$, and $S_{\max} = W + 1$ in Theorem~\ref{thm:algo} to obtain the output $z_u$ satisfying $|z_u-z^*_u|<\varepsilon$ with probability at least $\frac{2}{3}$. The query complexity is $O(\frac{\lVert b\rVert_{\infty}^2S_{\max}}{\delta^3\varepsilon^2}\log \frac{\lVert b\rVert_{\infty}}{\delta\varepsilon})=O(\frac{W+1}{\varepsilon^2}\log \frac{1}{\varepsilon})$. 
\end{proof}

\section{Omitted Proofs for Lower Bound}\label{sec:lb-pf}
The following lemma shows that the answers of two instances have a constant gap. 
\begin{lemma}\label{lem:diff}
For the random instances $(S^G,b^{G\text{-}0})$ and $(S^G,b^{G\text{-}1})$ and the random query vertex $u$ sampled uniformly at random from $G' \subseteq G$, the following holds with probability 1:
\begin{itemize}
    \item $z^*_u = 0$ where $z^*$ is the unique solution of $S^Gz = b^{G\text{-}0}$.
    \item $z^*_u \geq c_0$ where $z^*$ is the unique solution of $S^Gz = b^{G\text{-}1}$ and $c_0=\frac{1-e^{-\frac{1}{20d}}}{12(d+2)}$ is a constant.
\end{itemize}
\end{lemma}
The proof of \Cref{lem:diff} is given in Section \ref{sec:proof-diff}.

Next, we show that any sublinear-time algorithm with low query complexity cannot distinguish between the two instances $(S^G,b^{G\text{-}0})$ and $(S^G,b^{G\text{-}1})$ for $G \sim \mu_n$.
We first introduce the definition of query-answer history for an abstract sublinear algorithm.
Let $\mathcal{A}$ be a sublinear-time algorithm using  $\ell$ queries to a given $(S,b)$-oracle. A query-answer history $D^{\mathcal{A}}$ is an ordered tuple $[(q_0,a_0),(q_1,a_1),\cdots,(q_\ell,a_\ell)]$ such that 
\begin{itemize}
    \item $q_0 = \empty$ and $a_0$ is the input of the algorithm;
    \item For any $1\leq i \leq \ell$, $q_i$ is the $i_{th}$ query to the $(S,b)$-oracle and $a_i$ is the answer of $q_i$ returned by the oracle. For the $i_{th}$ query, a (possibly probabilistic) mapping maps $\{(q_j,a_j)\mid 0\leq j < i\}$ to the query $q_i$, where $q_i$ is one of the three types of queries in~\Cref{def:oracle}.
\end{itemize}
Finally, a (possibly probabilistic) mapping maps $\{(q_j,a_j)\mid 0\leq j < \ell\}$ to the output of $\mathcal{A}$. Abstractly, the algorithm can be viewed as first sampling $x \sim D^{\mathcal{A}}$, and then mapping (possibly probabilistic mapping) $x$ to the output $\mathcal{A}(x)$. 

Let $G \sim \mu_n$. Let $u$ be a uniform random vertex in $G' \subseteq G$.
Let $P_0$ be the query access oracle of $(S^G,b^{G\text{-}0})$.
Let $P_1$ be the query access oracle of $(S^G,b^{G\text{-}1})$.
Let $\mathcal{A}$ be a sublinear-time algorithm given $u$ as the 
input, where $u$ is a uniform random vertex in $G'$.
Given the access to an oracle $P \in \{P_0,P_1\}$, consider the algorithm $\mathcal{A}$ that can distinguish whether $P=P_0$ or $P=P_1$, where the algorithm outputs \textbf{Yes} if $P = P_0$ or outputs $\textbf{No}$ if $P=P_1$.
Define $D_0^{\mathcal{A}}$ and $D_1^{\mathcal{A}}$ as the query-answer histories induced by the interaction of $\mathcal{A}$ and $P_0$ and $P_1$ respectively. 

Next, we bound the total variation distance between $D_0^{\mathcal{A}}$ and $D_1^{\mathcal{A}}$. Note that both $D_0^{\mathcal{A}}$ and $D_1^{\mathcal{A}}$ are random variables whose randomness comes from the randomness of $G \sim \mu_n$, the uniform random vertex $u$ sampled from $G'$, and the independent randomness inside $\mathcal{A}$.
\begin{restatable}{lemma}{lem:tv}\label{lem:tv}
    The following result holds for any constant $\alpha<1/2$.
    For any algorithm $\mathcal{A}$ that asks $m \leq \alpha k$ queries, the total variation distance between $D_1^{\mathcal{A}}$ and $D_2^{\mathcal{A}}$ is at most $6\alpha$.
\end{restatable}

The proof of \Cref{lem:tv} is given in Section \ref{sec:proof-tv}. With Lemma \ref{lem:diff} and Lemma \ref{lem:tv}, we are ready to prove the lower bound. 

\begin{proof}[\textbf{Proof of Theorem~\ref{thm:lb}}]
By \Cref{obs:S}, our hard instance satisfies properties stated in the theorem.
    Let $\mathcal{A}$ be an algorithm with access to $(S,b)$-oracle that solves the $Sz = b$ linear system for $1$-diagonally dominant matrix $S \in \mathbb{R}^{n \times n}$ and Boolean vector $b \in \{0,1\}^n$ with error $\varepsilon_0=\frac{c_0}{4}$, where $c_0$ is the universal constant in \Cref{lem:diff}.
We can apply the algorithm $\mathcal{A}$ in our hard instances $(S^G,b^{G\text{-0}})$ and $(S^G,b^{G\text{-1}})$ with the input vertex $u$, where $G \sim \mu_n$, $u$ is a uniform random vertex in subgraph $G'$ of $G$.
By \Cref{lem:diff}, by looking at the value of $z_u$, we can use $\mathcal{A}$ to distinguish between two instances $(S^G,b^{G\text{-0}})$ and $(S^G,b^{G\text{-1}})$ with success probability $\frac{2}{3}$. Formally, we can modify the output of $\mathcal{A}$ such that $\mathcal{A}$ outputs $\textbf{Yes}$ if $z_u < \frac{c_0}{2}$ and $\mathcal{A}$ outputs $\textbf{No}$ if $z_u \geq \frac{c_0}{2}$. By our assumption on $\mathcal{A}$, we have
\begin{align}\label{eq:prob}
   \Pr_{x \sim D_0^{\mathcal{A}}} [\mathcal{A}(x) = \textbf{Yes}] \geq \frac{2}{3}, \quad\text{and}\quad \Pr_{x \sim D_1^{\mathcal{A}}}[\mathcal{A}(x)  = \textbf{Yes}]  < \frac{1}{3}.
\end{align}

On the other hand, 
by Lemma \ref{lem:tv} with $\alpha = \frac{1}{20}$, if the algorithm $\mathcal{A}$ asks $\leq \frac{1}{20}k$ queries, then 
\begin{align*}
    d_{\text{TV}}(D_1^{\mathcal{A}}, D_2^{\mathcal{A}}) \leq \frac{6}{20} < \frac{1}{3}.
\end{align*}
Using the data processing inequality, we have
\begin{align*}
    \left\vert \Pr_{x \sim D_1^{\mathcal{A}}} [\mathcal{A}(x) = \textbf{Yes}] - \Pr_{x \sim D_2^{\mathcal{A}}}[\mathcal{A}(x)  = \textbf{Yes}] \right\vert \leq d_{\text{TV}}(D_1^{\mathcal{A}}, D_2^{\mathcal{A}}) < \frac{1}{3}.
\end{align*}
This contradicts~\eqref{eq:prob}. Hence, the algorithm $\mathcal{A}$ must ask more than $\frac{1}{20}k$ queries. By our definition of the graph $G$, we have $\max_{i\in[n]}|S^G_{ii}|=\Theta(dk) = \Theta(k)$. Then we finish the proof.
\end{proof}

\begin{remark} \label{remark:edge}
If the oracle allows the edge query, i.e., the oracle returns the information of a given edge $e\in E$, we can slightly modify the construction of hard instances to have the same lower bound. We can replace the isolated vertices by $n/k-1$ copies of $G'$ and $B$. Let the original parts be $G'_1$ and $B'_1$ and the copies be $G'_2$, $B_2$, $\cdots$, $G'_{\frac{n}{k}}$, $B_{\frac{n}{k}}$. We connect $G'_i$ and $B'_i$ in the same way by sampling a vertex $w_{G'_i}$ in $G'_i$ and $w_{B_i}$ in $B_i$ uniformly at random and connecting $w_{G'_i}$ and $w_{B_i}$. For the resulting graph $G$, $b^{G\text{-}1}$ and $b^{G\text{-}0}$ are defined as follows:

\begin{align*}
\forall v \in [n], \quad   b^{G\text{-}1}_v &= \begin{cases}
        1 & \text{if } v \text{ is a vertex in } B_1, \\
        0 & \text{otherwise}.
    \end{cases}\\
\forall v \in [n], \quad   b^{G\text{-}0}_v &= 0.
\end{align*}

It can be verified that it is hard to distinguish two instances by sampling edges and we have the same lower bound above.
\end{remark}

\subsection{Separations between solutions to two instances}\label{sec:proof-diff}

In this section, we prove Lemma \ref{lem:diff} to show that there is a constant gap between the answers of two instances. By the definition of $b^{G\text{-}0}$, it's obvious that $z_u^*=0$ for $S^{G}z= b^{G\text{-}1}$. To prove the second part, we need some properties of simple random walk on \emph{unweighted} expanders. 
We first give the definitions related to simple random walk and list some useful properties, and then we use these properties to prove the second part of Lemma \ref{lem:diff}.

Given a \emph{unweighted} graph $G=(V,E)$ and a current vertex $v \in V$, in every step, the simple random walk moves to a neighbor of $v$ uniformly at random.
Given two vertices $u,v \in V$, the \emph{hitting time} $T_G(u,v)$ is the expected number of steps for a simple random walk starting from $u$ to first visit $v$.

\begin{restatable}{lemma}{lem:hitting}\cite[Corollary 3.3]{lovasz1993random}\label{lem:hitting}
Given an unweighted graph $G=(V,E)$, let $T_G(u,v)$ be the hitting time of a random walk from $u$ to $v$ and $\kappa_{G}(u,v)=T_G(u,v)+T_G(v,u)$ be the commute time,
\[
T_G(u,v)\leq\kappa_{G}(u,v)\leq\frac{2m}{\gamma_G}\left(\frac{1}{d(u)}+\frac{1}{d(v)}\right),
\]
where $d(u)$ and $d(v)$ denote the unweighted degree of $u$ and $v$ in $G$ respectively.
\end{restatable}

Let $W \in [0,1]^{V \times V}$ be the random walk matrix of $G$, where $W_{uv}=1/d(u)$ if $u$ and $v$ are connected by an edge and $W_{uv}=0$ otherwise. Let $\pi$ be a distribution over $v$ such that $\pi_v = \frac{d(v)}{2m}$, where $m$ is the number of edges in $G$. It is well known that $W$ is reversible with respect to $\pi$, i.e., $\pi_u W_{uv} = \pi_v W_{vu}$ for all $u,v \in V$. Hence, $\pi$ is the stationary distribution of $W$. We have the following mixing lemma for the simple random walk on expanders.

\begin{restatable}{lemma}{lem:mixing}\cite[Theorem 5.1]{lovasz1993random}\label{lem:mixing}
Given a $d$-regular graph $G=(V,E)$. For any integer $\ell \geq 0$ and any vertex $u \in V$,
    \[
    \lVert \mathds{1}^T_{u}W^\ell-\pi\rVert_{\infty}\leq (1-\gamma_G)^{\ell},
    \]
    where $\mathds{1}_u$ is the indicator column vector of $u$ and $\gamma_G$ is the spectral expansion of $G$.
\end{restatable}

Return to our hard instance, where we sample a \emph{weighted} graph $G$ from $\mu_n$. Recall that $G$ is a union of $G'$, $B$ and $C$, where $G'$ and $B$ are two copies of $G_k^{\textnormal{ex}}$ and $C$ is a set of $n-2k$ isolated vertices. Two expanders $G'$ and $B$ are connected by an edge $w_{G'} \in G'$ and $w_{B} \in B$. Recall that although $G$ is weighted, the weight on every edge is the same. Hence, when consider the simple random walk on $G$, we can simply treat $G$ as an unweighted graph.
\begin{restatable}{lemma}{lem:rw}\label{lem:rw}
For a random graph $G \sim \mu_n$, the following holds with probability $1$.
Consider the simple random walk on $G$ starting from $w_B$. With probability at least $1/3$, the random walk will first enter into $B$ and stay in $B$ for at least $k/10$ steps.
\end{restatable}
\begin{proof}
    Note that the degree of $w_B$ is $d+1$ and $d$ neighbors of $w_B$ are in $B$.
	Hence, with probability $\frac{d}{d+1}$, the first step walks into $B$. Suppose the first step walks into a vertex $v_1 \in B$. 
    Note that if the random walk escapes $B$, it must visit vertex $w_B$ again because $w_B$ is the cut vertex connecting $G'$ and $B$.
    We only need to prove that for any vertex $v_1 \in B$ such that $v_1 \neq w_B$, with probability at least $2/5$, the simple random walk in $G' \cup B$ starting from $v_1$ first hits $w_B$ in at least $k/10$ steps.
    Denote the simple random walk on $G' \cup B$ starting from $v_1$ as $\mathcal{W}_1$.
    Denote the random walk sequence on $G' \cup B$ starting from $v_1$ as $X_0=v_1,X_1,\cdots,X_\ell$, where $\ell=k/10$. We have 
    \begin{align*}
        \Pr_{\mathcal{W}_1}[\forall 1\leq i\leq \ell, X_i\neq w_B] = 1 - \sum_{i=1}^\ell \Pr_{\mathcal{W}_1}[X_i = w_B \land \forall 1\leq j<i, X_j\neq w_B].
    \end{align*}
    Let $B_i$ denote the event that $X_i = w_B$ and $\forall 1\leq j<i, X_j\neq w_B$. In the above equation, we need to analyze the probability of $B_i$ in the simple random walk $\mathcal{W}_1$ on $G' \cup B$ starting from $v_1$. Consider the simple random walk $\mathcal{W}_2$ on $G[B]$ starting from $v_1$, where $G[B]$ is the induced graph of $B$ in $G$.
    We show that for any $i\geq 1$,
    \begin{align}\label{eq:same-prob}
        \Pr_{\mathcal{W}_1}[B_i] = \Pr_{\mathcal{W}_2}[B_i].
    \end{align}
    To verify this, consider a coupling between $\mathcal{W}_1$ and $\mathcal{W}_2$ such that we can couple two walks perfectly up to the moment that the random walk $\mathcal{W}_1$ first hits the vertex $w_B$.
    Note that conditional on the current vertex is $w_B$, two walks follow different transition rules because $\mathcal{W}_1$ can move to $G'$ and $\mathcal{W}_2$ can only stay in $B$. Hence, after two walks meet at $w_B$, we couple two walks independently.
    Suppose the event $B_i$ occurs in $\mathcal{W}_1$. By the definition of $B_i$, up to the time $i$, $\mathcal{W}_1$ only can do the random walk inside $B$. For every transition in time $j$ with $j < i$, two walks follow the same transition rule. 
    By the coupling, two walks are coupled perfectly up to time $i$ and the event $B_i$ also occurs in $\mathcal{W}_2$. We have $\Pr_{\mathcal{W}_1}[B_i] \leq \Pr_{\mathcal{W}_2}[B_i]$. Similarly, if $B_i$ occurs in $\mathcal{W}_2$, then it also occurs in $\mathcal{W}_1$. We have $\Pr_{\mathcal{W}_2}[B_i] \leq \Pr_{\mathcal{W}_1}[B_i]$. Hence, we have $\Pr_{\mathcal{W}_1}[B_i] = \Pr_{\mathcal{W}_2}[B_i]$. 
    In the random walk $\mathcal{W}_2$, since the induced graph $G[B]$ is the $d$-regular graph in \Cref{prop:expander}, the stationary distribution of the simple random walk is the uniform distribution over $k$ vertices. Hence, 
    by \Cref{lem:mixing} and the expansion property of $B$,
	\begin{align*}
		\Pr_{\mathcal{W}_1}[\forall 1\leq i\leq \ell, X_i\neq w_B] 
		=& 1 - \sum_{i=1}^\ell \Pr_{\mathcal{W}_1}[B_i] = 1 - \sum_{i=1}^\ell \Pr_{\mathcal{W}_2}[B_i]\\
        = & 1 - \sum_{i=1}^\ell \Pr_{\mathcal{W}_2}[X_i = w_B \land \forall 1\leq j<i, X_j\neq w_B]\\
        \geq&  1 - \sum_{i=1}^\ell \Pr_{\mathcal{W}_2}[X_i = w_B]\\
	\text{(by \Cref{lem:mixing})}\quad	\geq& 1-\sum_{i=1}^\ell\left(\frac{1}{k}+(1-\gamma_{G_N^{\textnormal{ex}}})^i\right)\\
    \text{(by $\ell = \frac{k}{10}$)} \quad   \geq& \frac{9}{10}-\frac{1-\gamma_{G_N^{\textnormal{ex}}}}{\gamma_{G_N^{\textnormal{ex}}}}.
	\end{align*}
    Since $\gamma_{G_N^{\textnormal{ex}}}>2/3$, $\Pr_{\mathcal{W}_1}[\forall 1\leq i\leq \ell, X_i\neq w_B]\geq2/5$. 
    The lower bound holds for $X_0 = v_1$ for any $v_1 \in B$ with $v_1 \neq w_B$.
    Recall that, in the first step, the random walk enter into a vertex $v_1 \in B$ with probability $\frac{d}{d+1} \geq \frac{10}{11}$ as $d \geq 10$.
    So the final probability is at least $\frac{10}{11} \cdot \frac{2}{5}>\frac{1}{3}$.
\end{proof}

With Lemma \ref{lem:hitting}, Lemma \ref{lem:mixing} and Lemma \ref{lem:rw}, we have the proof of Lemma \ref{lem:diff}.

\begin{proof}[\textbf{Proof of Lemma~\ref{lem:diff}}]
    It is easy to get the first part that $z_u=0$ for $S^{G}z=0$, since $b_v^{G\text{-}0}=0$ for all $v\in[n]$.  To prove the second part, let $z^*$ denote the unique solution of $S^{G}z=b^{G\text{-}1}$.
    We use \Cref{alg:recursive} as a tool to analyze the solution $z^*$.

    Let $u$ be the query vertex sampled uniformly at random from $G'$.
    Image we run the algorithm \Cref{alg:recursive} on the query vertex $u$ with the query access to $(S,b)$-oracle, where $S = S^{G}$ and $b = b^{G\text{-}1}$.
    Let $z_u$ denote the random variable returned by RecursiveSolver($u$). By \Cref{lem:recursive},
    \begin{align*}
        \E[z_u] = z^*_u.
    \end{align*}
    By the definition of our hard instance,
    for any vertex $v$ in the subgraph $G' \cup B$, 
    \begin{itemize}
        \item If $v \notin \{w_{G'},w_B\}$, then the out degree $d_{\text{out}}(v) = d k$ and $S_{vv} = d_{\text{out}}(v) + 1$;
        \item If $v = w_{G'}$ or $v = w_B$, then the out degree $d_{\text{out}}(v) = (d+1) k$ and $S_{vv} = d_{\text{out}}(v)+ 1$;
    \end{itemize}
    For any edge $\{v,w\}$ in graph $G$, it holds that 
    \begin{align*}
        S_{vw} = - k.
    \end{align*}

    \Cref{alg:recursive} returns the value of $\text{RecursiveSolver}(u)$ for the query vertex $u$. Initially, it starts from $u$.
    In every step, suppose the current vertex is $v$, with probability $\frac{|S_{vv}| - \dout_v}{|S_{vv}|}$, it returns the value $\text{RecursiveSolver}(v) = \frac{\text{sgn}(S_{vv}) b_v}{|S_{vv}| - \dout_v}$; with rest probability, it perform a simple random walk to a random neighbor $w$ and returns \[\text{sgn}(-S_{vv}S_{vw})\cdot \text{RecursiveSolver}(v) = \text{RecursiveSolver}(v),\]
    where the equation holds because $S_{vv} > 0$ and $S_{vw} < 0$.
    The algorithm performs the simple random walk because every edge in $G$ has the same weight and the graph $G$ is simple.
    
    Hence, the whole algorithm can be abstracted as follows. It starts from $u$. Suppose the current vertex is $v$, with probability $\frac{|S_{vv}| - \dout_v}{|S_{vv}|} = \frac{1}{|S_{vv}|}$, the whole algorithm terminates at vertex $v$ and returns the value $\frac{\text{sgn}(S_{vv}) b_v}{|S_{vv}| - \dout_v} = \mathbf{1}[v \in B]$; with rest probability, it performs a simple random walk step and recurses.
    If we can show that \Cref{alg:recursive} terminates in $B$ with a constant probability $\Omega(1)$, then we have $z_u=\Omega(1)$ because $\frac{\text{sgn}(S_{vv}) b_v}{|S_{vv}| - \dout_v} = 1$ for all $v\in B$.

    First, we bound the probability $p_1$ that the algorithm touches the vertex $w_{G'}$.
    Let $\mathcal{W}_1$ denote the simple random walk on $G'\cup B$ starting from $u$.
    Let $\mathcal{W}_2$ denote the simple random walk on $G'$ starting from $u$.
    Let $A$ denote the event that the simple random walk hits $w_{G'}$ within $4k/\gamma_{G_N^{\textnormal{ex}}}$ steps. Similar to the proof of \eqref{eq:same-prob} in Lemma~\ref{lem:rw}, we have $\Pr_{\mathcal{W}_1}[A] = \Pr_{\mathcal{W}_2}[A]$.
    By Lemma \ref{lem:hitting}, the hitting time $T_{G'}(v,w_{G'})\leq 2k/\gamma_{G_N^{\textnormal{ex}}}$. Using the Markov inequality,
    \[
        \Pr_{\mathcal{W}_1}[A] = \Pr_{\mathcal{W}_2}[A] = 1 - \Pr_{\mathcal{W}_2} \left[ \bar{A} \right]\geq 1-\frac{T_{G'}(v,w_{G'})}{4k/\gamma_{G_N^{\textnormal{ex}}}}\geq \frac12.
    \]
    Since in each step, the algorithm terminates independently with probability $\frac{|S_{vv}| - \dout_v}{|S_{vv}|} = \frac{1}{|S_{vv}|} \leq \frac{1}{dk+1}$.
    The algorithm touches vertex $w_G'$ and does not terminate with probability at least
    \[
        p_1\geq \Pr_{\mathcal{W}_1}[A]\times\left(1-\frac{1}{dk + 1}\right)^{\frac{2k}{\gamma_{G_N^{\textnormal{ex}}}}}\geq\frac{1}{4},
    \]
    where the last inequality holds because $\gamma_{G_N^{\textnormal{ex}}} \geq \frac{2}{3}$ and $d \geq 10$.
    With constant probability $p_2=\left(1-\frac{1}{k(d+1)+1}\right)\frac{1}{d+1}\geq\frac{1}{d+2}$ algorithm walks into $w_B \in B$ from $w_{G'}$. Then we bound the probability $p_3$ that algorithm terminates in $B$ conditional on the simple random walk hitting $w_B$.  By Lemma \ref{lem:mixing}, with probability at least $1/3$, the simple random walk stays in $B$ for at least $k/10$ steps, in each step, the algorithm terminates and outputs 1 independently with probability $\frac{|S_{vv}| - \dout_v}{|S_{vv}|} = \frac{1}{|S_{vv}|} \geq \frac{1}{(d+1)k+1}$. We have
    \[
    p_3\geq \frac{1}{3}\times\left(1-\left(1-\frac{1}{k(d+1)+1}\right)^{k/10}\right)\geq \frac{1}{3}\left(1-e^{-\frac{1}{20d}}\right).
    \]
   By the chain rule, we have $p=p_1\cdot p_2\cdot p_3\geq\frac{1-e^{-\frac{1}{20d}}}{12(d+2)}$. So, we have $z_u^*\geq p\cdot 1\geq c_0$.
\end{proof}

\subsection{The total variation distance between two query-answer histories}\label{sec:proof-tv}

\begin{proof}[\textbf{Proof of Lemma~\ref{lem:tv}}]
    Fix an algorithm $\mathcal{A}$ that asks $m \leq \alpha k$ queries. 
    Our goal is to bound the total variation distance between two query-answer histories $\mathcal{D}_0^\mathcal{A}$ and $\mathcal{D}_1^\mathcal{A}$, where for $i \in \{0,1\}$, the randomness of  $\mathcal{D}_i^\mathcal{A}$ comes from the randomness of $(S^G,b^{G\text{-}i})$ for $G \sim \mu_n$, the uniform random input vertex $u$ from $G'$, and the independent randomness inside the algorithm $\mathcal{A}$. For $i \in \{0,1\}$, let \[\mathcal{D}^\mathcal{A}_i = \{(q^i_0,a^i_0),(q^i_1,a^i_1),\ldots,(q^i_m,a^i_m)\},\] 
    where $q^i_0 = \emptyset$, $a^i_0 = u$ is the input, $q^i_t$ is the $t_{\text{th}}$ query and $a^i_t$ is the $t_{\text{th}}$ answer by the oracle $P_i$. For $i \in \{0,1\}$, for each time $t\in[m]$, after interaction with the oracle $P_i$ for $t$ times, the algorithm $\mathcal{A}$ knows a knowledge graph $G_t^i$, which contains all vertices and edges that appeared in the query-answer history $\mathcal{D}_i^\mathcal{A}$ up to time $t$. Note that for some vertices $v$ and all edges $vw$ in $G_t^i$, the algorithm knows $ \Dout_v $, $ \dout_v $, $ S_{vv} $, and $ b_v $ at vertex $v$ and the weight $S_{vw}$ of edge $vw$.
    
    We construct a coupling $\mathcal{C}$ between $\mathcal{D}_0^\mathcal{A}$ and $\mathcal{D}_1^\mathcal{A}$ and show that 
    $\Pr_{(X,Y)\sim\mathcal{C}}[X \neq Y] \leq 6\alpha$. Using the coupling inequality, the total variation distance is at most 
    \begin{align*}
        d_{TV}(\mathcal{D}_0^\mathcal{A},\mathcal{D}_1^\mathcal{A}) \leq \Pr_{(X,Y)\sim\mathcal{C}}[X \neq Y] \leq 6\alpha.
    \end{align*}
    The coupling $\mathcal{C}$ is defined as follows. We first sample $G \sim \mu_n$ and sample $u$ from $G'$ uniformly at random. We use the same graph $G$ to construct $(S^G,b^{G\text{-}0})$ and $(S^G,b^{G\text{-}1})$. Let $P_i$ be the oracle that answers the query of the algorithm $\mathcal{A}$ with $(S^G,b^{G\text{-}i})$.
    In every step $t \in [m]$, we couple $(a^0_t,q^0_t)$ and $(a^1_t,q^1_t)$ to maximize the probability that $(a^0_t,q^0_t) = (a^1_t,q^1_t)$.
    Let $u$ be the input of the algorithm $\mathcal{A}$. For each $t \in [m]$, each $i \in \{0,1\}$, define the bad event $\mathcal{E}_t^i$
    \begin{itemize}
        \item $\mathcal{E}_t^i$: $\mathcal{A}$ knows the vertex $w_{G'}$ or some vertex $w \in B$ from the knowledge graph $G_t^i$.
    \end{itemize}
    For any $0\leq t \leq m$, define the prefix of query-answer history generated by the coupling $\mathcal{C}$ up to time $t$ as follows
    \begin{align*}
    X_t = \{(q^0_0,a^0_0),(q^0_1,a^0_1),\ldots,(q^0_t,a^0_t)\},\quad
    Y_t = \{(q^1_0,a^1_0),(q^1_1,a^1_1),\ldots,(q^1_t,a^1_t)\}.
    \end{align*}
    Note that $X = X_m$ and $Y = Y_m$ for $(X,Y)\sim\mathcal{C}$.
    Initially, $q^0_0 = q^1_0 = \emptyset$, $a^0_0 = a^1_0 = u$ so that $X_0 = Y_0$. The event $\mathcal{E}_0^0\lor\mathcal{E}_0^1$ occurs only if $u = w_{G'}$, which happens with probability $\frac{1}{k}$.
    Define 
    \begin{align*}
        q_0 = \Pr\left[X_0 \neq Y_0 \lor \mathcal{E}_0^0\lor\mathcal{E}_0^1\right] = \frac{1}{k}.
    \end{align*}
    For any $t \in [m]$, we define the following quantity
    \begin{align*}
     q_t = \Pr\left[X_t \neq Y_t \lor \mathcal{E}_t^0 \lor \mathcal{E}_t^1 \mid X_{t-1} = Y_{t-1} \land \overline{\mathcal{E}_{t-1}^0} \land \overline{\mathcal{E}_{t-1}^1}\right].
    \end{align*}
    We now upper bound the value of $q_t$ for $t \in [m]$.
    Let $\mathcal{E}$ denote the event  $X_t \neq Y_t \lor \mathcal{E}_t^0 \lor \mathcal{E}_t^1$.
    Note that two instances differ only at the value of $b^{G\text{-}0}_v$ and $b^{G\text{-}1}_v$ for $v \in B$.
    Given the condition $X_{t-1} = Y_{t-1} \land \overline{\mathcal{E}_{t-1}^0} \land \overline{\mathcal{E}_{t-1}^1}$, two knowledge graphs $G_{t-1}^0$ and $G_{t-1}^1$ up to time $t-1$ are the same. We use $G_{t-1}$ to denote knowledge graphs $G_{t-1}^0 = G_{t-1}^1$. 
     Given the condition, two queries $q^0_t$ and $q^1_t$ can be coupled perfectly because they are a (possibly probabilistic) function of the query-answer history up to time $t-1$.
    Consider three different kind of queries that the algorithm $\mathcal{A}$ can ask.
    \begin{itemize}
        \item Vertex query. If $\mathcal{A}$ performs a vertex query on a vertex $v$ in $G_{t-1}$, then since the event $\mathcal{E}_{t-1}^0 \lor \mathcal{E}_{t-1}^1$ does not occur, the answer $a^0_t$ and $a^1_t$ are the same and the event $\mathcal{E}$ cannot occur. Suppose $\mathcal{A}$ performs a vertex query on a vertex $v$ not in $G_{t-1}$. Note that all the labels of vertices are assigned as a uniform random permutation of $[n]$, conditional on the knowledge graph $G_{t-1}$, the probability that the sampled vertex falls into $B \cup \{w_{G'}\}$ is at most $\frac{k+1}{n-t}$ because $G_{t-1}$ contains at most $t$ vertices. The answers $a^0_t$ and $a^1_t$ are the same if $v \notin B$. Hence, the probability that the event $\mathcal{E}$ occurs is at most $q_t \leq \frac{k+1}{n-t}$. We remark that the queried vertex $v$ is computed from the current query-answer history. The vertex may be deterministic or randomized. The analysis above only uses the randomness of the labels of the vertices.
        \item Neighbor and random walk queries. In the neighbor query, the algorithm queries the oracle with $(v,i)$, where $v$ is a vertex and $i \in [\Dout_v]$. 
        In the random walk query, the algorithm queries the oracle with $v$, the oracle returns a uniform random neighbor of $v$ if the out-degree of $v$ is not 0, otherwise, the oracle returns $\perp$. 
        Again, the query can be either deterministic or randomized, which is computed from the current query-answer history. We analyze two types of queries in a unified way.

        Since the answer depends only on graph structure, two answers $a^0_t$ and $a^1_t$ can also be coupled perfectly. We only need to bound the probability of $\mathcal{E}^0_t \lor \mathcal{E}^1_t$. Here are two cases.
        \begin{enumerate}
            \item  If the vertex $v$ is in $G_{t-1}$, then the algorithm may learn a new neighbor of the current knowledge graph. Since $w_{G'}$ is sampled uniformly at random from $G'$ where $|G'|=k$, the probability that the algorithm learns $w_{G'}$ is at most $\frac{1}{k-t}$. Conditional on $\overline{\mathcal{E}^0_{t-1}} \land \overline{\mathcal{E}^1_{t-1}}$, the algorithm cannot learn some $w \in B$ because $w_{G'}$ and $B$ are not in the current knowledge graph $G_{t-1}$.
            \item  If $v$ is not in $G_{t-1}$, then the algorithm learns at most two adjacent vertices $w$ and $v$, where $w$ is a random neighbor of $v$. 
            Note that $w \in B$ only if $v \in \{w_{G'}\}\cup B$.
            Let $H = G' \setminus (\{w_{G'}\} \cup G_{t-1})$ denote all vertices in $G'$ except $w_{G'}$ and vertices in knowledge graph $G_{t-1}$.
            Since we assume $v$ not in $G_{t-1}$, $w = w_{G'}$ occurs only if $v \in H$.
            Hence, the event $\mathcal{E}^0_t \lor \mathcal{E}^1_t$ occurs only if $v \in B \cup \{w_{G'}\}$ or $v \in H \land w = w_{G'}$.
            Since all the labels of vertices are assigned as a uniform random permutation of $[n]$, the probability that $v \in B \cup \{w_{G'}\}$ is at most $\frac{k+1}{n-t}$.
            Suppose $G_{t-1}$ contains $x$ vertices where $1\leq x \leq t$.
            The probability that $v \in H$ is $\frac{k-x-1}{n-x}$. Let us further condition on $v \in H$, which means the algorithm learns one more vertex in $G'$ and the algorithm still does not know $w_{G'}$.
            The vertex $w$ is another new vertex in $G'$ that is adjacent to $v$.
            Since $w_{G'}$ is sampled uniformly at random from $G'$, $w_{G'}$ is uniformly distributed among remaining vertices in $G'$. Hence, the probability that $w = w_{G'}$ is $\frac{1}{k-x-1}$.
            Hence, if $v$ is not a vertex in $G_{t-1}$, the probability of $\mathcal{E}_t^0 \lor \mathcal{E}_t^1$ is at most $\frac{k+1}{n-t}+\frac{k-x-1}{n-x}\cdot \frac{1}{k-x-1} = \frac{k+1}{n-t}+\frac{1}{n-x} \leq \frac{k+2}{n-t}$.
        \end{enumerate}
    \end{itemize}
    Combining all cases together, we have
    \begin{align*}
        q_t \leq \max\left\{ \frac{k+1}{n-t}, \frac{1}{k-t}, \frac{k+2}{n-t} \right\}\leq \frac{k+2}{n-t}+\frac{1}{k-t}.
    \end{align*}

    Finally, note that $X_0 = Y_0$, if $X_m \neq Y_m$, then we can find the first $0\leq t \leq m$ such that the event $X_t \neq Y_t \lor \mathcal{E}_t^0 \lor \mathcal{E}_t^1$ occurs. 
    Note that the value of $t$ depends on whether $X_t \neq Y_t$ or $\mathcal{E}_t^0 \lor \mathcal{E}_t^1$ first occurs. Such $t$ must exist because $X_t \neq Y_t$ must occur for some $t$.
    Taking a union bound over all possible $0\leq t \leq m$, we have
    \begin{align*}
        \Pr[X_m \neq Y_m] \leq \sum_{t=0}^m q_t \leq \frac{1}{k} + \sum_{t=1}^m \left( \frac{k+2}{n-t} + \frac{1}{k-t} \right) \leq \frac{1}{k} + \frac{(k+2)m}{n-m} + \frac{m}{k-m}.
    \end{align*}
    Note that $m = \alpha k$, $k < \sqrt{n}$, and $\alpha \leq \frac{1}{2}$ is a constant. For sufficiently large  $n$, the above bound is
    \begin{align*}
        \Pr[X_m \neq Y_m] \leq \frac{1}{k} +\frac{km}{n-m} + \frac{2m}{n-m} + \frac{\alpha}{1-\alpha} \leq (1+o(1))\alpha + o(1) + \frac{\alpha}{1-\alpha} \leq 6\alpha.
    \end{align*}
    This finishes the proof.
\end{proof}

\end{document}